\documentclass{./cta-author}

\newtheorem{theorem}{Theorem}{}
{}
\newtheorem{remark}{Remark}{}
\usepackage[inline,shortlabels]{enumitem}
\usepackage{amssymb,amsmath,amstext,amsfonts}
\usepackage{mathtools}
\usepackage{changepage}
\usepackage{lipsum}

\usepackage{siunitx}
\usepackage{physics}
\usepackage{multicol}
\usepackage[toc,page]{appendix}

\usepackage{float}
\usepackage{framed} % Framing content
\usepackage{multicol} % Multiple columns environment
\usepackage{nomencl} % Nomenclature package
\makenomenclature

\usepackage{booktabs}

\usepackage[caption=false,subrefformat=parens]{subfig}
\graphicspath{%
	{Figures/}
	{Figures/Matlab/}
	{Figures/MSWord/}
}

\usepackage[firstpageonly=true]{draftwatermark}				

\usepackage{hyperref}
\hypersetup{
    colorlinks=true,
    linkcolor=blue,
    filecolor=magenta,      
    urlcolor=cyan,
    bookmarks=true,
}
\definecolor{darkblue}{rgb}{0.0, 0.0, 0.55}

\begin{document}

\DraftwatermarkOptions{%
angle=0,
hpos=0.5\paperwidth,
vpos=0.985\paperheight,
fontsize=0.012\paperwidth,
color={[gray]{0.2}},
text={
  % Disclaimer as in
  % https://digital-library.theiet.org/files/Author_self-archiving_policy.pdf
  % https://v2.sherpa.ac.uk/romeo/
  \parbox{0.99\textwidth}{This paper is a postprint of a paper submitted to and accepted for publication in IET Control Theory \& Applications and is subject to Institution of Engineering and Technology Copyright. The copy of record is available at the IET Digital Library. (doi: \href{http://dx.doi.org/10.1049/iet-cta.2018.6163}{10.1049/iet-cta.2018.6163})}},
}

%\supertitle{Submission Template for IET Research Journal Papers}

\title{Online Model-Free Reinforcement Learning for the Automatic Control of a Flexible Wing Aircraft}

\author{\au{Mohammed Abouheaf$^{1,2\corr}$}, \au{Wail Gueaieb$^{1}$}, and \au{Frank Lewis$^{3}$}}

\address{\add{1}{School of Electrical Engineering and Computer Science, University of Ottawa, Ottawa, Ontario, Canada}
	\add{2}{Electrical Engineering, College of Energy Engineering, Aswan University, Aswan, Egypt}
	\add{3}{Electrical Engineering, University of Texas at Arlington, Arlington, Texas, United States}
	\email{mohammed.abouheaf@uottawa.ca}}

\begin{abstract}
The control problem of the flexible wing aircraft is challenging due to the prevailing and high nonlinear deformations in the flexible wing system. This urged for new control mechanisms that are robust to the real-time variations in the wing's aerodynamics. An online control mechanism based on a value iteration reinforcement learning process is developed for flexible wing aerial structures. It employs a model-free control policy framework and a guaranteed convergent adaptive learning architecture to solve the system's Bellman optimality equation. A Riccati equation is derived and shown to be equivalent to solving the underlying Bellman equation. The online reinforcement learning solution is implemented using means of an adaptive-critic mechanism. The controller is proven to be asymptotically stable in the Lyapunov sense. It is assessed through computer simulations and its superior performance is demonstrated on two scenarios under different operating conditions.
\end{abstract}

\maketitle

\section{Introduction} 
\label{sec:Introduction}
Unmanned flexible wing aircraft have gained an increasing attention due to their low operating cost, long endurance, and relatively short runway. Recent research investigations tackled several aspects of flexible wing systems, including modeling approaches and control mechanisms. Unlike fixed wing aircraft, the high nonlinear deformation in the flexible wing makes it very challenging to model such type of aircraft and to design autonomous controllers for them~\cite{cook_spottiswoode_2005,Kilkenny_1983,Kilkenny_1984,Cook_Kilkenny_1986,Kilkenny_1986,Blake_1991,cook_1994,Ochi_2017,Sweeting1981}. Classical control approaches depend on accurate system models, and so may not be suitable for these unmanned aerial vehicles (UAVs). In this paper, an innovative model-free adaptive learning controller is developed for the automatic real-time control of flexible wing aircraft. This work brings together ideas form the machine learning, optimal control, and adaptive critics.

A significant body of fundamental research has already been conducted on flexible wing aircraft~\cite{Kilkenny_1983,Sweeting1981,Kilkenny_1984,Blake_1991,Kilkenny_1986,Cook_Kilkenny_1986,cook_1994,Cook_2013}. The early studies of hang gliders involved aeroelastic and aerodynamic modeling approaches, which resulted in approximate aerodynamic equations of motion~\cite{Kroo_1983}. A mobile test rig was developed in~\cite{Kilkenny_1983} to test the aerodynamic characteristics of a set of hang glider wings. In addition, a fundamental database of wind tunnel tests for hang gliders is provided in~\cite{Kilkenny_1984}. Powton~\cite{Powton_1995} studied the effects of the aerodynamic nonlinearities in terms of the wing's camber and twist variations. A study of the aerodynamic models for flexible wing aircraft along with a practical measurement setup was conducted in~\cite{cook_spottiswoode_2005}. 
Furthermore, a number of researchers attempted to decouple the aircraft's motion along the longitudinal and lateral planes. For example, Kroo~\cite{Kroo_1983} developed longitudinal and lateral small perturbation equations of motion in~\cite{Kroo_1983}. A longitudinal-lateral aerodynamic model is usually derived by assuming a rigid wing while the aerodynamic derivatives are added at the end of the modeling process~\cite{DE_MATTEIS_1990,De_Matteis_1991}. Another of such models was developed by Spottiswoode~\cite{Spottiswoode_2001}. 

The dynamics of a flexible wing system are usually modeled by considering two interacting bodies (the wing and the pilot/fuselage frame) under geometric and kinematic constraints~\cite{Ochi_2015}. A nine-degree-of-freedom (DOF) flight dynamical model has been recently proposed by Ochi~\cite{Ochi_2017} with a pilot-handling feedback control mechanism. The work followed a rigorous analytical process which led to a set of nonlinear state equations.

One of the most commonly methods of controlling a flexible wing aircraft is to use a weight shift mechanism~\cite{Kilkenny_1983,Kilkenny_1984,Cook_Kilkenny_1986}. In this case, the motion of the aircraft is governed by the shift in the relative centers of gravity of the wing and the fuselage (as opposed to the fuselage weight shift itself)~\cite{cook_spottiswoode_2005}. It is shown that the geometry of the control arms affects the the maximum control moments in a flexible wing system~\cite{cook_spottiswoode_2005} and that the static pitching stability can be increased by lowering the center of gravity. A mathematical model is employed in~\cite{cook_spottiswoode_2005} to determine the stability of a flexible wing vehicle within a speed envelope. It is shown that the lateral stability margins are larger compared to those of a typical airplane.  The initial investigations about the longitudinal stability of the hang glider were conducted in~\cite{Blake_1991}. The classical longitudinal stability is shown to be effective in understanding the robustness of the hang glider provided that some conditions are satisfied~\cite{cook_1994}. The effect of the wing's aerodynamic nonlinearities on the pitching moment is studied in~\cite{Powton_1995}. The longitudinal stability and control properties of the flexible wing UAV are investigated in~\cite{Rollins_2000}. More stability analyses were conducted in~\cite{DE_MATTEIS_1990,De_Matteis_1991} but in the frequency domain this time. A regulation system combined of two control structures for planar vertical takeoff and landing aircraft is proposed in~\cite{Aguilar17}. The vertical variable is regulated using saturated-feedback linearization mechanism, while the angular and horizontal variables are controlled using combined sliding mode-PD control mechanism. Regulator schemes are proposed for magnetic generator and hexarotor mechanical systems where the stability is analyzed using  Lyapunov mathematical framework~\cite{Rubio18}. Another controller based on a Lyapunov method and robust feedback linearization scheme is proposed for nonlinear processes in~\cite{RUBIO2018155}.

Artificial Intelligence (AI) has become a vivid area to solve various optimal control problems for single and multi-agent systems. Dynamic programming is typically applied to solve optimization problems in the framework of optimal control theory~\cite{Howard_1960,Werbos1992}. However, dynamic programming methods generally suffer from the curse of dimensionality in both the state and action spaces. Hence, Approximate Dynamic Approaches (ADP) are developed to alleviate these difficulties by using Reinforcement Learning (RL) approaches and means of adaptive critics to provide temporal difference solutions for the optimal control problems~\cite{Howard_1960,Werbos1992,Bert_1995,Werbos1990,Sutton_1998}. The ADP approaches can be classified into four main categories:
\begin{enumerate*}[(i), itemjoin={{; }}, itemjoin*={{; and }}]
\item Heuristic Dynamic Programming
\item Dual Heuristic Dynamic Programming
\item Action Dependent Heuristic Dynamic Programming 
\item Action Dependent Dual Heuristic Dynamic Programming.
\end{enumerate*}
These categories depend on the form of the solving value function and the accompanying choice of the optimal control actions. Different versions of approximate dynamic programming solutions are proposed to solve multi-agent control problems in~\cite{AbouheafRV2017,AbouheaIJCNN2013,AbouheafCH2014}. These include a class of innovative implementations for the respective costate equations of the underlying control problems. Classical mechanics are used to derive the basic optimality principles to solve for optimal control problems in~\cite{Lewis_2012,Bellman1957,Bryson1996}. The minimum cost-to-go function is shown to be equivalent to the solution of the underlying Hamilton-Jacobi-Bellman equation (HJB)~\cite{Lewis_2012,Bellman1957,Bryson1996}. The optimal control theory provides the platform to derive the optimality conditions to solve various types of control problems using the ADP approaches~\cite{Lewis_2012,Howard_1960}. These conditions involve expressions for the optimal value functions and the respective optimal control policies, and hence result in the underlying Bellman optimality, Hamilton-Jacobi-Bellman (HJB), and costate equations.

Reinforcement Learning (RL) approaches optimize the performance of the dynamical systems by minimizing the respective infinite horizon cumulative sum of a cost function~\cite{Sutton_1998,Sen1999}. They are developed using two-step processes referred to as either value iteration or policy iteration, depending on the strategy followed to solve the underlying problem~\cite{Bert_1995,Werbos1990,Sutton_1998}. The main differences between the two are in the structures of the evaluated value functions, which in turn affects the update of the optimal policies. This is equivalent to solving the Hamilton-Jacobi-Bellman equation for the considered dynamical system.
An online policy iteration approach is used to solve dynamic graphical games in~\cite{AbouheafRV2017}. Some RL solutions are implemented using means of the adaptive critics in~\cite{Sutton_1998}. The implementations were carried out using separate actor-critic neural network structures, where the actor approximates the optimal policy, while the critic approximates the optimal value function~\cite{Widrow1973}. The analytical methodology for the critic learning process was presented in~\cite{Widrow1973}. The actor-critic learning mechanism is implemented forward-in-time and is accomplished via a dynamic learning environment, where the quality of the taken actions are evaluated using a utility function~\cite{Sutton_1998}. The adaptive critics are used to implement some optimal control problem solutions in~\cite{Werbos1989,Werbos1992,Werbos1974}. More RL approaches are used to solve a number of multi-agent optimization problems in game theoretic frameworks in~\cite{Busoniu2008,Vrancx2008,Tamimi2008,Vrabie2009}.

This work contributes with an innovative computational and mathematical framework to design model-free control schemes for a class of nonlinear processes with unknown dynamical models like flexible wing systems. The proposed control scheme has the ability to provide model-free optimal control decisions while remaining robust to the unmodeled uncertainties and disturbances. The adaptive learning scheme is able to solve the underlying model-free Bellman optimality equation (temporal difference equation) of the dynamical system in real-time. The dynamic learning system is realized using value iteration process with two supporting separate neural network structures, where gradient descent approach is used to tune the neural networks weights. A Riccati development is introduced to understand the duality between the model-based control solution and its equivalent model-free temporal difference solution.

The paper is organized as follows:
Section~\ref{sec:system-model} briefly introduces the weight shift control mechanism of a flexible wing aircraft and defines its state space representation.
Section~\ref{sec:heur-dynam-progr} discusses the optimal control formulation of the problem along with its Heuristic Dynamic Programming solution.
In Section~\ref{sec:model-free-control}, we lay out the mathematical foundation for the model-free optimal control development, leading to the online adaptive learning algorithm along with its convergence proof.
The duality between the developed model-free Bellman optimality equation and the associated Riccati formulation is explained in Section~\ref{sec:ricc-contr-solut}.
The proposed online adaptive control scheme is based on an adaptive-critic technique that is implemented through two artificial neural networks. Section~\ref{sec:adapt-crit-solut} details the design of these neural networks. 
Section~\ref{sec:results-discussion} presents different simulation scenarios to demonstrate the salient features of the proposed controller.
Finally, the manuscript is concluded with a few concluding remarks in Section~\ref{sec:conclusion}.

\section{Flexible Wing Aircraft Model}
\label{sec:system-model}
The basic principles of operation of a flexible wing aircraft are briefly explained here. The framework adopted in this manuscript is based on decoupling the flight dynamics of a flexible wing system into longitudinal and lateral frames of motion. This is based on the findings in~\cite{cook_1994,cook_spottiswoode_2005}. The resulting dynamical models are used later on to validate the performance of the adaptive learning control algorithm.

The flexible wing aircraft is characterized by the mechanical coupling of two bodies: a wing and a fuselage which is suspended below the wing through a pivot (hang strap). A schematic diagram of a flexible wing system is shown in Figure~\ref{fig:fwa}. The coupling between the two structures imposes a set of dynamical and kinematic constraints on the system. The flexible wing aircraft employs an interesting pitch-roll control mechanism by pitching and rolling the control bar. These angles are the only two signals used to control the hang glider. They are denoted by $\alpha$ (roll) and $\beta$ (pitch) in Figure~\ref{fig:fwa}. Such pitch and roll movements on the control bar lead to a relative shift between the fuselage and the wing's centers of gravity. This creates aerodynamic moments which act on the aircraft's center of gravity. The inherent flexibility of the wing results in high nonlinear aerodynamics. The dynamics of the aircraft are referred to the wing's frames of motion (longitudinal and lateral). The approximate flight dynamics are modeled at trim speeds in~\cite{cook_1994,cook_spottiswoode_2005}. The classical control approaches require a precise mathematical model of the system, which is difficult to obtain for systems with this magnitude of complexity.

\begin{figure}[htb]
	\centering
	\includegraphics[width=0.99\columnwidth]{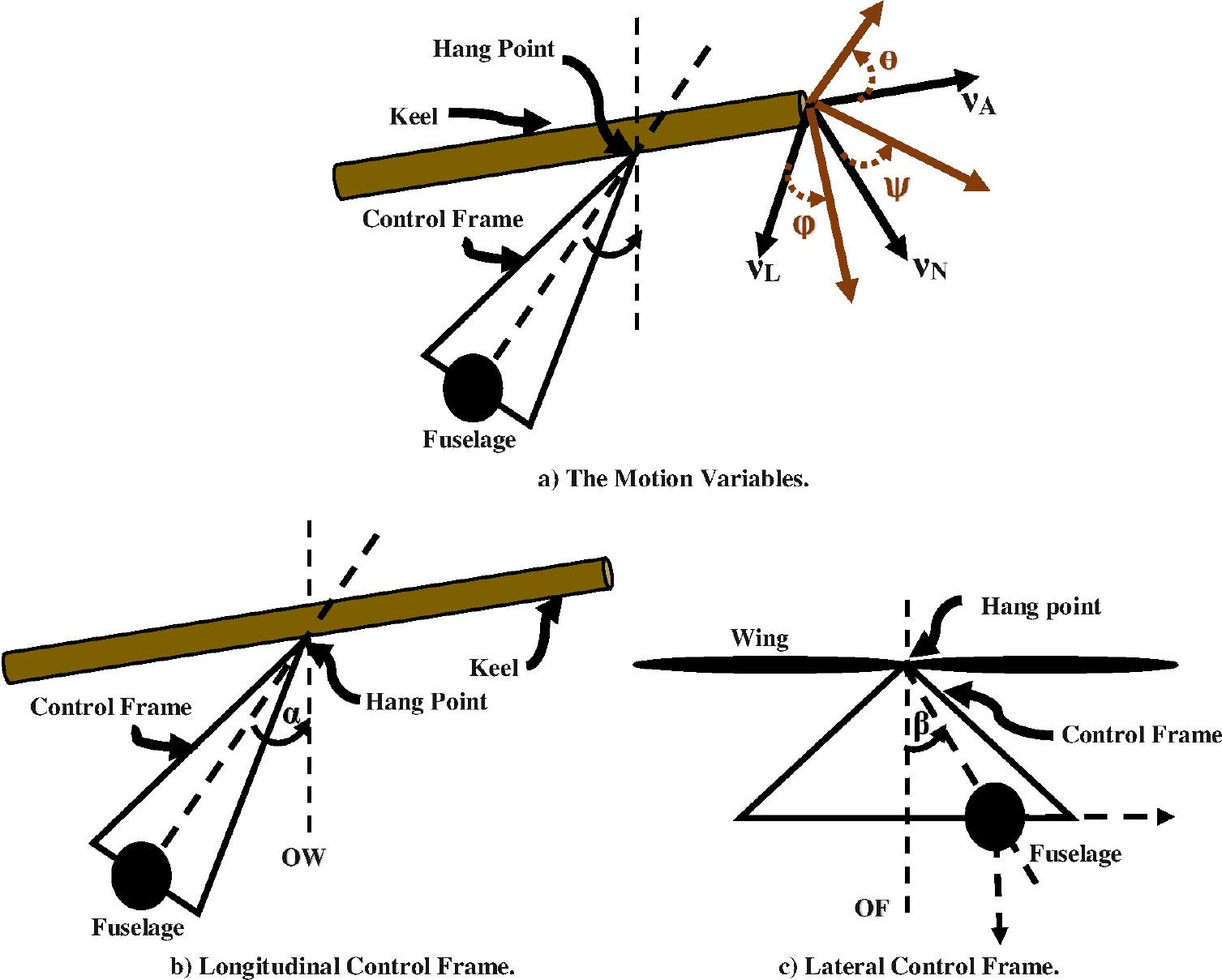}
	\caption{Schematic Diagram of the Flexible Wing System.}
	\label{fig:fwa}
\end{figure}

The following discrete-time linear model is developed for a flexible wing aircraft based on the nonlinear state space representation and the semi-experimental study conducted in~\cite{cook_spottiswoode_2005}:
\begin{eqnarray}
\boldsymbol{Z}_{k+1} = \boldsymbol{A} \boldsymbol{Z}_k  +  \boldsymbol{B}  \boldsymbol{u}_k
\label{E_02}
\end{eqnarray}
where the states along the longitudinal and lateral frames are
$\boldsymbol{Z}^{Lon} = [ \nu_A ~ \nu_N ~ \dot \theta ~ \theta ]^T$
and
$\boldsymbol{Z}^{Lat} = [ \nu_L ~ \dot \phi ~ \dot \psi ~ \phi ~ \psi ]^T$, respectively. 
The states $\theta$, $\phi$, and $\psi$, denote the glider's pitch, roll, and yaw attitude angles. The dot on top of a signal symbolizes its time derivative. The variables $\nu_A, \nu_N,$ and $\nu_L$ represent the axial, normal, and lateral linear velocities. The control signals of the longitudinal and lateral subsystems are $\boldsymbol{u}^{Lon}=\alpha$ and $\boldsymbol{u}^{Lat}=\beta$, respectively.

It is worth pointing out that this model does not account for the interaction forces between the fuselage and the wing. Also, it models the fuselage's reaction control moment as a spring-damper system. More information about the assumptions and limitations of the model can be found in~\cite{cook_spottiswoode_2005}. It is also important to articulate the fact that the above model is only used for simulating the hang glider's dynamical behavior under the control strategy to be detailed later. The latter, however, is model independent.

\section{Heuristic Dynamic Programming Solution}
\label{sec:heur-dynam-progr}
This section, shows the Heuristic Dynamic Programming solution for the underlying optimal control problem of the flexible wing aircraft. This solution is implemented online using a value iteration process, where the solution depends partially on the dynamics of the flexible wing system. The following development is used to frame the model-free adaptive learning controller.

Consider a system described by the following discrete-time state space mode:
\begin{equation}
\boldsymbol{Z}_{(k+1)}=\boldsymbol{A} \, \boldsymbol{Z}_{k} + \boldsymbol{B} \, \boldsymbol{u}_{k}, 
\label{dyn}
\end{equation}
where $k$ is the time index, $\boldsymbol{Z}\in \mathbb{R}^{n}$ is a vector of the states, $\boldsymbol{u}\in \mathbb{R}^{m}$ is a vector of the control input signals, and $\boldsymbol{A}$ and $\boldsymbol{B}$ are the system and input matrices, respectively.
A performance measure index $J = \sum_{k=0}^{\infty} U(\boldsymbol{Z}_{k},\boldsymbol{u}_{k})$ is employed to measure the quality of the taken control policies using a quadratic convex cost function $U(\boldsymbol{Z}_{k},\boldsymbol{u}_{k})= \frac{1}{2} \left(\boldsymbol{Z}^\mathrm{T}_{k} \, \boldsymbol{Q} \, \boldsymbol{Z}_{k}+ \boldsymbol{u}^\mathrm{T}_{k} \, \boldsymbol{R} \, \boldsymbol{u}_{k}\right)$, 
and 
where $\boldsymbol{R} > 0 \in \mathbb{R}^{m \times m}$ and $\boldsymbol{Q} \ge 0 \in \mathbb{R}^{n \times n}$ are positive and semi-positive definite matrices that are symmetric. Let $S(\boldsymbol{Z}_{k})$ be a value function that is quadratic in the states $\boldsymbol{Z}_{\ell}$, such that
\begin{equation}
  \label{valx}
  S(\boldsymbol{Z}_{k}) = \sum_{r=k}^{\infty} \, U(\boldsymbol{Z}_{r},\boldsymbol{u}_{r})
\end{equation}
It is possible to define Bellman equation as follows:
\begin{eqnarray}
S(\boldsymbol{Z}_{k})&=& \frac{1}{2} (\boldsymbol{Z}^\mathrm{T}_{k} \, \boldsymbol{Q} \, \boldsymbol{Z}_{k}+ \boldsymbol{u}^\mathrm{T}_{k} \, \boldsymbol{R} \, \boldsymbol{u}_{k} ) +S(\boldsymbol{Z}_{(k+1)}).
\label{bellx}
\end{eqnarray}

The optimal policies are evaluated using the Bellman optimality principles, which result in the stationarity condition (the optimal control policy $u^o$)~\cite{Lewis_2012,AbouheafCH2014,AbouheafRV2017}. This is achieved by applying the gradient of~(\ref{bellx}) with respect to the control signal $\boldsymbol{u}_k$, $S^o(\boldsymbol{Z}_{k}) = \text{argmin}_{\boldsymbol{u}_{k}} S(\boldsymbol{Z}_{k})$, which yields
\begin{equation}
\boldsymbol{u}_{k}^o=- \boldsymbol{R}^{-1} \boldsymbol{B}^\mathrm{T} \, \nabla S^o(\boldsymbol{Z}_{(k+1)}),
\label{optx}
\end{equation}
where $\nabla S^o(\boldsymbol{Z}_{k})=\partial S^o (\boldsymbol{Z}_{k}) \, / \,\partial \boldsymbol{Z}_{k}$.
Plugging the optimal policy~(\ref{optx}) in~(\ref{bellx}) leads to the Bellman optimality equation (defining the optimal value function $S^o (\dots)$),
\begin{eqnarray}
{S}^{o}\big(\boldsymbol{Z}_{k}\big)
= \frac{1}{2}\Big(\boldsymbol{Z}_{k}^{\mathrm{T}} \, \boldsymbol{Q} \, 
\boldsymbol{Z}_{k}+\boldsymbol{u}_{k}^{o{\mathrm{T}}} \, \boldsymbol{R} \, \boldsymbol{u}_{k}^{o}\Big) + {S}^{o} \big(\boldsymbol{Z}_{(k+1)}\big).
\label{obellx}
\end{eqnarray}
The solution of the Bellman optimality equation~(\ref{obellx}) can be implemented either offline or online. A value iteration process is employed here to solve the underlying Heuristic Dynamic Programming Problem as follows:
\begin{center}
  \textbf{Algorithm 1: HDP-Value Iteration Solution} 
\end{center}
\begin{enumerate}
\item Initialize $S^0(\boldsymbol{Z}_{k})$ and $\boldsymbol{u}_{k}^0$.
\item Evaluate the value $S^{(\ell+1)}(\dots)$
  \begin{eqnarray}
    S^{(\ell+1)}(\boldsymbol{Z}_{k})  = U^\ell(\boldsymbol{Z}_{k},\boldsymbol{u}_{k})+S^{\ell}(\boldsymbol{Z}_{(k+1)}),
    \label{algbell}
  \end{eqnarray}
  where $\ell$ is the iteration index.
\item Update the control action
  \begin{eqnarray}
    \footnotesize
    \boldsymbol{u}_{k}^{(\ell+1)}=- \boldsymbol{R}^{-1} \boldsymbol{B}^{\mathrm{T}} \, \nabla S^{(\ell+1)}(\boldsymbol{Z}_{(k+1)}).
    \label{algopt}
  \end{eqnarray}
\item Terminate on convergence of $\Vert S^{\ell+1}(..) - S^{\ell}(..)\Vert$.
\end{enumerate}

The weighting matrix $R$ is a positive definite matrix, which guarantees that~(\ref{optx}) and consequently~(\ref{algopt}) have non-singular values. Later on, the optimal control policy will not be a function of $R$.

This algorithm represents an online solution for the optimal control problem. However, the optimal control action~(\ref{algopt}) uses partial knowledge about the  system's dynamics (i.e., $\boldsymbol{u}^{\ell}$ depends on the input control matrix $\boldsymbol{B}$).
In the following, we will introduce a control mechanism that is able to produce model-free policies.

\section{Model-Free Control Mechanism}
\label{sec:model-free-control}
Herein, an online model-free control algorithm is proposed. First, a modified form of Bellman optimality equation~(\ref{obellx}) is obtained with a model-free policy structure. Then, a convergence analysis of the resultant adaptive learning controller is carried out.

\subsection{Modified Bellman Equation}
Consider the modified form of the value function~(\ref{valx}), which is
quadratic in the states $\boldsymbol{Z}_k$ and the control signals $u_k$. This yields the following modified Bellman equation:
\begin{eqnarray}
S(\boldsymbol{Z}_{k},\boldsymbol{u}_{k})=U(\boldsymbol{Z}_{k},\boldsymbol{u}_{k})+ S(\boldsymbol{Z}_{(k+1)},\boldsymbol{u}_{(k+1)}).
\label{bellxu}
\end{eqnarray}
Neglecting the high-order terms of Taylor's series expansion of the value function~\eqref{bellxu}, $S(\dots)$ can be approximated by
\begin{align}
  S(\boldsymbol{Z}_{k},\boldsymbol{u}_{k})
  & \approx  \frac{1}{2}
    \begin{bmatrix*}[c]
      \boldsymbol{Z}_{k} ^\mathrm{T}  & \boldsymbol{u}_{k}^\mathrm{T}
    \end{bmatrix*}
    % [\boldsymbol{Z}_{k} ^\mathrm{T}  \boldsymbol{u}_{k}^\mathrm{T}]
                                        \boldsymbol{M}
    \begin{bmatrix*}[c]
      \boldsymbol{Z}_{k}\\ \boldsymbol{u}_{k} 
    \end{bmatrix*} ~,
  \label{E_15}
\end{align}
where $\boldsymbol{M}
=
\begin{bmatrix*}[l]
\boldsymbol{M}_{\boldsymbol{Z}\boldsymbol{Z}} & \boldsymbol{M}_{\boldsymbol{Z}\boldsymbol{u}} \\
\boldsymbol{M}_{\boldsymbol{u} \boldsymbol{Z}}&\boldsymbol{M}_{\boldsymbol{u}\boldsymbol{u}}
\end{bmatrix*} >0 \, \in \mathbb{R}^{(n+m) \times (n+m)}.$

Similarly, the Bellman's optimality principles are applied to the difference equation~(\ref{bellxu}), or equivalently the value function~(\ref{E_15}), to obtain the optimal policy $\boldsymbol{u}_{k}^o$ and the optimal value $S^o(\dots)$.
\begin{equation}
  \label{oconxu}
  \boldsymbol{u}_{k}^o
  = \text{argmin}_{\boldsymbol{u}_{k}} S(\boldsymbol{Z}_{k},\boldsymbol{u}_{k})
  = -  \boldsymbol{M}_{ \boldsymbol{u}_{k}\boldsymbol{u}_{k}}^{-1} \cdot \boldsymbol{M}_{ \boldsymbol{u}_{k}\boldsymbol{Z}_{k}} \cdot \boldsymbol{Z}_{k}.
\end{equation}

	\begin{remark}
	The optimal policy (\ref{oconxu}) relies on the successive approximations of value function $S(\dots)$ using Bellman temporal difference equation (\ref{bellxu}) (i.e., evaluating matrix $M$) as will be explained in the following online model-free value iteration process. Later on, the convergence analysis emphasizes the positive definiteness of the obtained solutions towards the optimal solution starting with initial positive definite value function (i.e., $M^0 \, > \, 0$ ). This guarantees that the optimal policies will not have any singular values since the block matrix $\boldsymbol{M}_{\boldsymbol{u}\boldsymbol{u}}$ is positive definite at all evaluation steps.
	\end{remark}

Substituting (\ref{oconxu}) in (\ref{bellxu}) yields a model-free of Bellman optimality equation,
\begin{equation}
  \label{obellxy}
  {S}^{o}\big(\boldsymbol{Z}_{k}, \boldsymbol{u}_{k}\big)
  = \frac{1}{2}\Big(\boldsymbol{Z}_{k}^{\mathrm{T}} \boldsymbol{Q}
  \boldsymbol{Z}_{k}+\boldsymbol{u}_{k}^{o{\mathrm{T}}}  \boldsymbol{R}  \boldsymbol{u}_{k}^{o}\Big) + {S}^{o} \big(\boldsymbol{Z}_{(k+1)}, \boldsymbol{u}^o_{ (k+1)}\big).
\end{equation}

Based on this solution, a model-free control algorithm is developed to solve for the optimal value function $S^{o}(\dots)$ and the optimal strategy $\boldsymbol{u}^o$, evaluated by (\ref{oconxu}), as listed below.
\begin{center}
  \textbf{Algorithm~2: Value Iteration Control Algorithm} 
\end{center}
\begin{enumerate}
\item Initialize $S^0(\boldsymbol{Z}_{k},\boldsymbol{u}_{k}^0)$ and $\boldsymbol{u}_{k}^0$ with admissible values.
  %\vspace{5pt}
\item Update the value function $S^{(\ell+1)}(\ldots)$
  \begin{eqnarray}
    S^{(\ell+1)}(\boldsymbol{Z}_{k},\boldsymbol{u}_{k})  = U^\ell(\boldsymbol{Z}_{k},\boldsymbol{u}_{k}) +S^{\ell}(\boldsymbol{Z}_{(k+1)},\boldsymbol{u}_{(k+1)}),
    \label{abellxu}
  \end{eqnarray}
  where $\ell$ is the iteration index.
\item {Update the policy} 
  \begin{eqnarray} 
    u_{(k+1)}^{(\ell+1)}&=&- {\boldsymbol{M}}^{-1(\ell+1)}_{\boldsymbol{u} \boldsymbol{u}} {\boldsymbol{M}}^{(\ell+1)}_{\boldsymbol{u} \boldsymbol{Z}} \boldsymbol{Z}_{(k+1)}^{\ell+1}.
                          \label{aoconxu}
  \end{eqnarray}
\item Terminate on convergence of $\Vert S^{\ell+1}(..) - S^{\ell}(..)\Vert$.
\end{enumerate}

\subsection{Value Iteration Policy Convergence}
For the controller to be useful it is important to have guaranteed convergence property. This can be proven by studying the evolution of the value functions $S^\ell(\dots)$ during the online learning process of Algorithm~2. 
\begin{theorem}
  \label{thm:stability-convergence}
  Consider the dynamic system~\eqref{dyn} with the adaptive control law of Algorithm~2. Then,
  \begin{enumerate}
  \item The equilibrium point $(\boldsymbol{Z}_{k},\boldsymbol{u}_{k})=(\boldsymbol{0},\boldsymbol{0})$ is asymptotically stable.
  \item The sequence of the value function $S^{\ell}(\dots)$, for $\ell=0,1,\dots$, is monotonically ascending and converges to the optimal solution of~(\ref{obellxy}).
  \item The sequence of policies~(\ref{aoconxu}) is stabilizing.
  \end{enumerate}
\end{theorem}
\begin{proof}
\begin{enumerate}
\item
The solving value function $S(\dots)$, which takes the form (\ref{E_15}), is a Lyapunov candidate function. This function is continuously differentiable with respect to $\boldsymbol{Z}$ and $\boldsymbol{u}$, where  $S(\boldsymbol{Z}_{k}=\boldsymbol{0},\boldsymbol{u}_{k}=\boldsymbol{0})=0$,
 $S(\boldsymbol{Z}_{k},\boldsymbol{u}_{k})=\sum_{r=k}^{\infty} \, U(\boldsymbol{Z}_{r},\boldsymbol{u}_{r}) > 0$, $\forall (\boldsymbol{Z}_{k},\boldsymbol{u}_{k}) \neq (\boldsymbol{0},\boldsymbol{0})$.

The policy~(\ref{aoconxu}) is stationary with the following value
$$S(\boldsymbol{Z}_{k})=\boldsymbol{Z}_{k}^T \left(\boldsymbol{M}_{\boldsymbol{Z}_{k}\boldsymbol{Z}_{k}}-\boldsymbol{M}_{\boldsymbol{Z}_{k}\boldsymbol{u}_{k}} \boldsymbol{M}_{\boldsymbol{u}_{k}\boldsymbol{u}_{k}}^{-1} \boldsymbol{M}_{\boldsymbol{u}_{k}\boldsymbol{Z}_{k}} \right) \boldsymbol{Z}_{k}>0.$$
Schur complement of the symmetric matrix $\boldsymbol{M}_{\boldsymbol{u}_{k}\boldsymbol{u}_{k}}$ is positive definite such that  $\left(\boldsymbol{M}_{\boldsymbol{Z}_{k}\boldsymbol{Z}_{k}}-\boldsymbol{M}_{\boldsymbol{Z}_{k}\boldsymbol{u}_{k}} \boldsymbol{M}_{\boldsymbol{u}_{k}\boldsymbol{u}_{k}}^{-1} \boldsymbol{M}_{\boldsymbol{u}_{k}\boldsymbol{Z}_{k}} \right)>0$. This means that, block matrices $\boldsymbol{M}_{\boldsymbol{Z}_{k}\boldsymbol{Z}_{k}}$ and $\boldsymbol{M}_{\boldsymbol{u}_{k}\boldsymbol{u}_{k}}$ are also positive definite. Hence, the hessians of $S(\boldsymbol{Z}_k,\boldsymbol{u}_k)$ with respect to  $\boldsymbol{Z}_k$ and $\boldsymbol{u}_k$ are given by $\displaystyle \frac{\partial^2 S(\boldsymbol{Z}_k,\boldsymbol{u}_uk)}{\partial \boldsymbol{Z}^2_k}=\boldsymbol{M}_{\boldsymbol{Z}_{k}\boldsymbol{Z}_{k}}>0$ and $\displaystyle \frac{\partial^2 S(\boldsymbol{Z}_k,\boldsymbol{u}_k)}{\partial \boldsymbol{u}^2_k}=\boldsymbol{M}_{\boldsymbol{u}_{k}\boldsymbol{u}_{k}}>0$ respectively. Therefore, $(\boldsymbol{0},\boldsymbol{0})$ is an equilibrium point for the convex function $S(\dots)$.  Since the following holds 
\begin{eqnarray}
S(\boldsymbol{Z}_{(k+1)},\boldsymbol{u}_{(k+1)})-S(\boldsymbol{Z}_{k},\boldsymbol{u}_{k}) = \nonumber \\ -U\big(\boldsymbol{Z}_{k},\boldsymbol{u}_{k}\big) < 0,\forall (\boldsymbol{Z}_{k},\boldsymbol{u}_{k}) \neq (\boldsymbol{0},\boldsymbol{0}).
\end{eqnarray}
Therefore,  $S(..)$ is a Lyapunov function and the equilibrium point $(\boldsymbol{0},\boldsymbol{0})$ is asymptotically stable.
\item
Let the initial value function $S^0 \geq 0$ be upper bounded above by a constant $C$. Applying an arbitrary policy $\boldsymbol{\mu}$ yields 
\begin{equation}
0 \le S^0(\boldsymbol{Z}_{k},\boldsymbol{\tilde  u}) \le S^0(\boldsymbol{Z}_{k},\boldsymbol{\mu}_{k}) \le C,
\label{ins}
\end{equation}
where $ \boldsymbol{\tilde u}$ is a policy satisfying~(\ref{aoconxu}).

Using a stabilizing policy $\boldsymbol{u}$ and the iterative process (\ref{abellxu}) we get
\begin{multline*}
  S^{\ell+1}\big(\boldsymbol{Z}_{k},\boldsymbol{u}\big)-S^{\ell}\big(\boldsymbol{Z}_{k},\boldsymbol{u}\big)=S^{\ell}\big(\boldsymbol{Z}_{(k+1)},\boldsymbol{u}\big) \\
  +\dfrac{1}{2}\Big(\boldsymbol{Z}_{k}^{\mathrm{T}}\boldsymbol{Q}\boldsymbol{Z}_{k}+\boldsymbol{u}^{{\mathrm{T}}} \boldsymbol{R}
  \boldsymbol{u}\Big)-S^{\ell-1}\big(\boldsymbol{Z}_{(k+1)},\boldsymbol{u}\big) \\
  -\dfrac{1}{2}\Big(\boldsymbol{Z}_{k}^{\mathrm{T}}\boldsymbol{Q}\boldsymbol{Z}_{k}
+\boldsymbol{u}^{{\mathrm{T}}}\boldsymbol{R}\boldsymbol{u}\Big) .
\end{multline*}
This equality can be extended into the following arrangement:
\begin{align*}
&S^{\ell+1}\big(\boldsymbol{Z}_{k},\boldsymbol{u}\big)-S^{\ell}\big(\boldsymbol{Z}_{k},\boldsymbol{u}\big) \\
&=
S^{\ell}\big(\boldsymbol{Z}_{(k+1)},\boldsymbol{u}\big)-S^{\ell-1}\big(\boldsymbol{Z}_{(k+1)},\boldsymbol{u}\big) \\
&=
S^{(\ell-1)}\big(\boldsymbol{Z}_{(k+2)},\boldsymbol{u}\big)-S^{\ell-2}\big(\boldsymbol{Z}_{(k+2)},\boldsymbol{u}\big) \\
&=
S^{(\ell-2)}\big(\boldsymbol{Z}_{(k+3)},\boldsymbol{u}\big)-S^{\ell-3}\big(\boldsymbol{Z}_{(k+3)},\boldsymbol{u}\big) \\
&=\ \ \ \ \vdots \\
&= S^{1}\big(\boldsymbol{Z}_{(k+\ell)}, \boldsymbol{u}\big)
-S^{0}\big(\boldsymbol{Z}_{(k+\ell)}, \boldsymbol{u}\big).
\end{align*}
Rearranging this equality in a finite summation form yields
\begin{align*}
  &S ^{\ell+1}\big(\boldsymbol{Z}_ {k},\boldsymbol{u} \big) \\
  &=S ^{1}\big(\boldsymbol{Z}_{(k+\ell)}, \boldsymbol{u}\big)
-S ^{0}\big(\boldsymbol{Z}_{(k+\ell)}, \boldsymbol{u}\big)
+S ^{\ell}\big(\boldsymbol{Z}_ {k},\boldsymbol{u} \big) \\
&=S ^{1}\big(\boldsymbol{Z}_{(k+\ell)}, \boldsymbol{u}\big) -S ^{0}\big(\boldsymbol{Z}_{(k+\ell)}, \boldsymbol{u}\big) \\
  &+S ^{1}\big(\boldsymbol{Z}_{(k+\ell-1)}, \boldsymbol{u}\big)-S ^{0}\big(\boldsymbol{Z}_{(k+\ell-1)}, \boldsymbol{u}\big)
+S ^{\ell-1}\big(\boldsymbol{Z}_ {k},\boldsymbol{u} \big) \\
&=S ^{1}\big(\boldsymbol{Z}_{(k+\ell)}, \boldsymbol{u}\big)
-S ^{0}\big(\boldsymbol{Z}_{(k+\ell)}, \boldsymbol{u}\big) \\  &+S ^{1}\big(\boldsymbol{Z}_{(k+\ell-1)}, \boldsymbol{u}\big)
-S ^{0}\big(\boldsymbol{Z}_{(k+\ell-1)}, \boldsymbol{u}\big) \\ &+S ^{1}\big(\boldsymbol{Z}_{(k+\ell-2)}, \boldsymbol{u}\big)-S ^{0}\big(\boldsymbol{Z}_{(k+\ell-2)}, \boldsymbol{u}\big)
+S ^{\ell-2}\big(\boldsymbol{Z}_ {k},\boldsymbol{u} \big) \\
&=\sum_{n=0}^{\ell}S ^{1}\big(\boldsymbol{Z}_{(k+n)},\boldsymbol{u} \big)
-\sum_{n=0}^{\ell} S ^{0}\big(\boldsymbol{Z}_{(k+n)},\boldsymbol{u} \big)+S ^{0}\big(\boldsymbol{Z}_ {k},\boldsymbol{u} \big) \\
&=\sum_{n=0}^{\ell}S ^{1}\big(\boldsymbol{Z}_{(k+n)},\boldsymbol{u} \big)
-\sum_{n=1}^{\ell} S ^{0}\big(\boldsymbol{Z}_{(k+n)},\boldsymbol{u} \big).
\end{align*}
Using the utility function $U$ in this summation yields
\begin{align}
S ^{\ell+1}\big(\boldsymbol{Z}_ {k},\boldsymbol{u} \big)
=&\
S ^{0}\big(\boldsymbol{Z}_{(k+\ell+1)},\boldsymbol{u} \big) \nonumber \\
+&\
\sum_{n=0}^{\ell}\dfrac{1}{2}\Big(\boldsymbol{Z}_{(k+n)}^{\mathrm{T}} \, \boldsymbol{Q} \, \boldsymbol{Z}_{(k+n)}+\boldsymbol{u} ^{{\mathrm{T}}} \, \boldsymbol{R} \, \boldsymbol{u} \Big).
  \nonumber \\
   \leq&\
S ^{0}\big(\boldsymbol{Z}_{(k+\ell+1)},\boldsymbol{u} \big) \nonumber \\
  +&\ \sum_{n=0}^{\infty}\dfrac{1}{2} \Big(\boldsymbol{Z}_{(k+n)}^{\mathrm{T}} \boldsymbol{Q}  \boldsymbol{Z}_{(k+n)} + \boldsymbol{u} ^{{\mathrm{T}}} \boldsymbol{R}  \boldsymbol{u} \Big).
     \label{sum}
\end{align}

Using the stability result of part~1, we deduce that the summation $\sum_{n=0}^{\infty} U\big(\boldsymbol{Z}_{(k+n)},\boldsymbol{u} \big)$ is upper bounded by a constant $\bar C$. Then, (\ref{ins}) and~(\ref{sum}) lead to,
$0 \leq S^{\ell+1}\big(\boldsymbol{Z}_ {k},\boldsymbol{u} \big) \leq C+\bar C$, $\forall k,\ell$.
Therefore,
\begin{align} 
0 \le \dots \le S ^{0} \le S ^{1} \le \dots \le S^{\ell} \dots \le S ^{*} = \lim_{\ell\to\infty} S^{\ell}.
\label{ineq}
\end{align}
Hence, the sequence $S^{\ell}$, $\ell=0,1,\ldots$, converges to an optimal value function $S^*$ which is the optimal solution of the modified model-free Bellman optimality equation~(\ref{obellxy}).
\item
  Inequality~(\ref{ineq}) and the sequence of policies~(\ref{aoconxu}) lead to
\begin{multline} 
  0 \le S ^{0}(\boldsymbol{Z}_ {k},\boldsymbol{u}^0 ) \le S ^{1}(\boldsymbol{Z}_ {k},\boldsymbol{u}^1 ) \le \dots \\
  \le S ^{\ell}(\boldsymbol{Z}_ {k},\boldsymbol{u}^\ell ) \dots \le S ^{*}(\boldsymbol{Z}_ {k},\boldsymbol{u}^* ) .
\label{ineq1}
\end{multline}
Thus,
\begin{equation}
0 \le S^\ell(\boldsymbol{Z}_ {k},\boldsymbol{u}^o_ {k}) \le S^\ell(\boldsymbol{Z}_ {k},\boldsymbol{u}_ {k}), \forall \ell.
\label{ineq2}
\end{equation}
for some arbitrary stabilizing policy $\boldsymbol{u}$ and optimal policy $\boldsymbol{u}^o$.

Inequalities~(\ref{ineq1}) and~(\ref{ineq2}) guarantee that the sequence of policies  $\boldsymbol{u}^\ell_ {k}$, $\ell=0,1,\ldots$, is stabilizing and convergent to the optimal policy $\boldsymbol{u}^* = - \boldsymbol{M}_{\boldsymbol{u} \boldsymbol{u} }^{-1*} \cdot \boldsymbol{M}^*_{\boldsymbol{u} \boldsymbol{Z} } \cdot \boldsymbol{Z}_ {k}$, where
$S^\ast(\boldsymbol{Z}_{k},\boldsymbol{u}_{k}) = \frac{1}{2} [\boldsymbol{Z}_{k} ^\mathrm{T} ~  \boldsymbol{u}_{k}^\mathrm{T}] \boldsymbol{M}^\ast 
\begin{bmatrix}
\boldsymbol{Z}_{k}\\
\boldsymbol{u}_{k}
\end{bmatrix}$
is the optimal value function and 
$\boldsymbol{M}^\ast =
\begin{bmatrix*} 
\boldsymbol{M}^\ast_{\boldsymbol{Z}\boldsymbol{Z}} & \boldsymbol{M}^\ast_{\boldsymbol{Z}\boldsymbol{u}}\\
\boldsymbol{M}^\ast_{\boldsymbol{u} \boldsymbol{Z}}&\boldsymbol{M}^\ast_{\boldsymbol{u}\boldsymbol{u}}
\end{bmatrix*}$
is the optimal gain.
\end{enumerate}
\end{proof}

\section{Riccati Control Solution}
\label{sec:ricc-contr-solut}
The next development lays out the mathematical foundation for the Riccati solution for the underlying optimal control problem. This solution is equivalent to solving ``recursively'' the modified Bellman optimality equation~(\ref{obellxy}) using the model-free policies~(\ref{oconxu}). 
\begin{theorem}
Let the value $S(\boldsymbol{Z}_{k},\boldsymbol{u}_{k})= \frac{1}{2} [\boldsymbol{Z}_{k} ^T ~ \boldsymbol{u}_{k}^T] \cdot \boldsymbol{\Psi} \cdot 
\begin{bmatrix}
\boldsymbol{Z}_{k}\\
\boldsymbol{u}_{k}
\end{bmatrix}$ be the value iteration function used to solve the Bellman optimality condition~(\ref{obellxy}) following policies~(\ref{oconxu}), where $\boldsymbol{\Psi}$ is a Riccati gain matrix of the following form:
\begin{equation*}
  \boldsymbol{\Psi} =
  \begin{bmatrix}
    \boldsymbol{\Psi}_{\boldsymbol{Z}\boldsymbol{Z}}&\boldsymbol{\Psi}_{\boldsymbol{Z}\boldsymbol{u}}\\
    \boldsymbol{\Psi}_{\boldsymbol{u} \boldsymbol{Z}}&\boldsymbol{\Psi}_{\boldsymbol{u}\boldsymbol{u}}
  \end{bmatrix}
\end{equation*}
Also, let
$\boldsymbol{\hat\Psi}= {\boldsymbol{\Psi}}^{-1}_{\boldsymbol{u} \boldsymbol{u}} {\boldsymbol{\Psi}}_{\boldsymbol{u} \boldsymbol{Z}}$
and
$\boldsymbol{\tilde \Psi}= \boldsymbol{\Psi}_{\boldsymbol{Z} \boldsymbol{Z}}- {\boldsymbol{\Psi}}^{-1}_{\boldsymbol{u} \boldsymbol{u}}  {\boldsymbol{\Psi}}_{\boldsymbol{u} \boldsymbol{Z}}$.
Then, the Riccati solution can be written in the following recursive expression:
\begin{eqnarray}
  \boldsymbol{\Psi}^{\ell+1} = \left[\begin{array}{ll}
\boldsymbol{A}^{\mathrm{T}}  \boldsymbol{\tilde \Psi}^{\ell} \boldsymbol{A}+ \boldsymbol{\hat\Psi}^{\ell\mathrm{T}} \, \boldsymbol{R} \,  \boldsymbol{\hat\Psi}^\ell+\boldsymbol{Q}  \qquad & \boldsymbol{A}^\mathrm{T}  \boldsymbol{\tilde\Psi}^\ell \boldsymbol{B} \\
\boldsymbol{B}^\mathrm{T}  \boldsymbol{\tilde\Psi}^\ell \boldsymbol{A}  & \boldsymbol{B}^\mathrm{T}  \boldsymbol{\tilde \Psi}^\ell \boldsymbol{B}
\end{array}\right].
\label{ric}
\end{eqnarray}
\end{theorem}
\begin{proof}
Using the above definition of the value function $S(\dots)$, the optimal policy may be expressed as
\begin{equation}
\boldsymbol{u}_{k}^o=-\, {\boldsymbol{\Psi}}^{-1}_{\boldsymbol{u} \boldsymbol{u}} \, {\boldsymbol{\Psi}}_{\boldsymbol{u} \boldsymbol{Z}} \, \boldsymbol{Z}_{k}=- \boldsymbol{\hat \Psi} \boldsymbol{Z}_{k} .
\label{locpolicy}
\end{equation}
The value of function $S(\dots)$ at time-step $k$ is 
\begin{equation}
{S} \big(\boldsymbol{Z}_{(k+1)}, \boldsymbol{u}_{(k+1)}\big)= \frac{1}{2} \boldsymbol{Z}^\mathrm{T}_{(k+1)}  \boldsymbol{\tilde\Psi} \boldsymbol{Z}_{(k+1)} .
\label{SL1}
\end{equation}
Using~(\ref{dyn}), (\ref{locpolicy}), and~(\ref{SL1}) in~(\ref{obellxy}), yields
\begin{multline}\label{E_31}
{S}\big(\boldsymbol{Z}_{k}, \boldsymbol{u}_{k}\big)
= \frac{1}{2}\Big(\boldsymbol{Z}_{k}^{\mathrm{T}} \, \boldsymbol{Q} \, \boldsymbol{Z}_{k} + \boldsymbol{Z}_{k}^{{\mathrm{T}}}  \boldsymbol{\hat\Psi}^\mathrm{T} \, \boldsymbol{R} \,  \boldsymbol{\hat\Psi} \, \boldsymbol{Z}_{k} \Big)  \\
+  \frac{1}{2} \left( \boldsymbol{Z}^\mathrm{T}_{k}\boldsymbol{A}^\mathrm{T}  \boldsymbol{\tilde\Psi} \boldsymbol{A} \boldsymbol{Z}_{k}+\boldsymbol{u}^\mathrm{T}_{k}\boldsymbol{B}^\mathrm{T}  \boldsymbol{\tilde\Psi} \boldsymbol{B} \boldsymbol{u}_{k} \right. \\
+ \left. \boldsymbol{u}^\mathrm{T}_{k} \boldsymbol{B}^\mathrm{T}  \boldsymbol{\tilde\Psi} \boldsymbol{A} \boldsymbol{Z}_{k}+ \boldsymbol{Z}^\mathrm{T}_{k}\boldsymbol{A}^\mathrm{T} \boldsymbol{\tilde \Psi} \boldsymbol{B} \boldsymbol{u}_{k}\right) .
\end{multline}
This equation can be rearranged as
\begin{multline*}
[\boldsymbol{Z}_{k} ^T ~  \boldsymbol{u}_{k}^T]  \boldsymbol{\Psi} 
\begin{bmatrix}
\boldsymbol{Z}_{k}\\
\boldsymbol{u}_{k}
\end{bmatrix}
= \\
[\boldsymbol{Z}_{k} ^T ~  \boldsymbol{u}_{k}^T]
\begin{bmatrix*}[l]
  \boldsymbol{Q}+\boldsymbol{\hat \Psi}^\mathrm{T} \, \boldsymbol{R} \,  \boldsymbol{\hat\Psi}+\boldsymbol{A}^\mathrm{T}  \boldsymbol{\tilde\Psi} \boldsymbol{A}  & \boldsymbol{A}^\mathrm{T}  \boldsymbol{\tilde\Psi} \boldsymbol{B} \\
  \boldsymbol{B}^\mathrm{T}  \boldsymbol{\tilde\Psi} \boldsymbol{A}  & \boldsymbol{B}^\mathrm{T}  \boldsymbol{\tilde\Psi} \boldsymbol{B}
\end{bmatrix*}
\begin{bmatrix}
\boldsymbol{Z}_{k}\\
\boldsymbol{u}_{k}
\end{bmatrix}.
%\label{E_s3}
\end{multline*}
Therefore, a Riccati equation can be formulated as
\begin{equation*}
  \boldsymbol{\Psi} =
  \begin{bmatrix*}[l]
          \boldsymbol{Q}+ \boldsymbol{\hat\Psi}^\mathrm{T} \, \boldsymbol{R} \,  \boldsymbol{\hat\Psi}+\boldsymbol{A}^\mathrm{T}  \boldsymbol{\tilde\Psi} \boldsymbol{A}  & \boldsymbol{A}^\mathrm{T}  \boldsymbol{\tilde\Psi} \boldsymbol{B}\\
          \boldsymbol{B}^\mathrm{T} \boldsymbol{\tilde \Psi} \boldsymbol{A}  & \boldsymbol{B}^\mathrm{T}  \boldsymbol{\tilde\Psi} \boldsymbol{B}
        \end{bmatrix*} .
%\label{E_s4}
\end{equation*}
This solution form is equivalent to solving the underlying modified Bellman equation~(\ref{obellxy}) using the model-free policies~(\ref{oconxu}), which results in the recursive solution form~(\ref{ric}).
\end{proof}

\section{Actor-Critic Design}
\label{sec:adapt-crit-solut}
In order to apply Algorithm~2 to solve the model-free optimality equation~(\ref{obellxy}) using the model-free policies~(\ref{oconxu}), it is necessary to approximate the value function ${S^o} \big(\boldsymbol{Z}_{k}, \boldsymbol{u}^o_{k}\big)$ and the optimal policy $\boldsymbol{u}^o_k$. This is accomplished through an actor-critic neural network structure. The critic approximates the value function while the actor approximates the optimal policy. Each of the two structures is implemented as a single-layer feedforward neural network where the weights are tuned in real-time using a gradient descent approach.

Since the model-free optimal control policy is a function of ${\boldsymbol{Z}}_k$, a policy may be approximated by the actor's output as
$\boldsymbol{\hat u}_{k}= {\boldsymbol{W}}_a^\mathrm{T} {\boldsymbol{Z}}_k$,
where the matrix ${\boldsymbol{W}}_a^\mathrm{T} \in \mathbb{R}^{n \times m}$ represents the weights of the actor neural network.
The value function is quadratic in the states and the model-free control policy. Then, a similar critic structure can be proposed, such that
$\hat S({\boldsymbol{Z}}_{k},\boldsymbol{\hat u}_{k})= \frac{1}{2} {\boldsymbol{\chi}}^\mathrm{T}_k  {\boldsymbol{W}}_c^\mathrm{T}  {\boldsymbol{\chi}}_k$,
where the matrix ${\boldsymbol{W}}_c^\mathrm{T} \in \mathbb{R}^{(n+m) \times (n+m)}$ represents the weights of the critic's neural network and ${\boldsymbol{\chi}}^\mathrm{T}_k=[{\boldsymbol{Z}}_{k} ^\mathrm{T} ~  \boldsymbol{\hat u}_{k}^\mathrm{T}]$.

The actor's approximation error is
${\delta}_{\boldsymbol{W}_a}=\boldsymbol{\hat u}_{k}-\boldsymbol{u}_{k}^{target}$,
where the target value of the mode-free optimal policy is
\begin{equation}
\boldsymbol{u}_{k}^{target} = -  \left[(\boldsymbol{W}^T_{c\boldsymbol{\hat u}_{k}\boldsymbol{\hat u}_{k}})^{-1} ~ {{\boldsymbol{W}}_c}^T_{\boldsymbol{\hat u}_{k}{\boldsymbol{Z}}_{k}}\right] \cdot  {\boldsymbol{Z}}_{k} .
\label{E_udesired}
\end{equation}
Following a gradient descent approach to tune the weights leads to the following weight adaptation law:
\begin{equation*}
{\boldsymbol{W}}_a^{(\ell+1)\mathrm{T}} = {\boldsymbol{W}}_a^{(\ell)\mathrm{T}}-\eta_{{W}_a}  \left({ \boldsymbol{\hat u}}^{(\ell)}_{k}- \boldsymbol{u}_{k}^{(\ell)desired}\right)  {\boldsymbol{Z}}^\mathrm{T}_k ,
%\label{E_33}
\end{equation*}
where $0<\eta_{{W}_a}<1$ is an actor's learning rate.

Similarly, the critic's approximation error is $\delta_{{W}_c}=\hat S_{k}-S_k^{target}$,
where the target value function is
\begin{equation*}
S_k^{target}= \frac{1}{2}\left({\boldsymbol{Z}}_{k}^{\mathrm{T}}  {\boldsymbol{Q}} {\boldsymbol{Z}}_{k} + \boldsymbol{\hat u}^\mathrm{T}_k  {\boldsymbol{R}} \boldsymbol{\hat u}_k\right)+\hat S({\boldsymbol{Z}}_{(k+1)},\boldsymbol{\hat u}_{(k+1)}) .
% \nonumber
% \label{E_36}
\end{equation*}
Differentiating with respect to the approximation weights and applying the chain rule yields the following adaptive law:
\begin{equation}
{\boldsymbol{W}}_c^{(\ell+1)\mathrm{T}}= {\boldsymbol{W}}_c^{(\ell)\mathrm{T}}-\eta_{{W}_c} \left(\hat S^{(\ell)}_{k}-S_k^{(\ell)desired}
\right) \boldsymbol{\chi}_k \boldsymbol{\chi}_k^\mathrm{T} ,
\label{E_38}
\end{equation}
with $0<\eta_{{W}_c}<1$ being a critic's learning rate.

The following algorithm summarizes the online neural network implementation of the adaptive actor-critic technique. 
\begin{center}
  \textbf{Algorithm~3: Neural Network Implementation of the Actor-Critic Technique}
    \begin{enumerate}
    \item Initialize the actor and critic weights, ${\boldsymbol{W}}_a^0$ and ${\boldsymbol{W}}_c^0$.
      
    \item Calculate the approximate model-free policy $\boldsymbol{\hat u}_k^\ell$ using~(\ref{E_31}). 
      
    \item Apply~(\ref{E_udesired}) to update the actor weights.
      
    \item Calculate the measurements ${ \boldsymbol{\hat Z}}_{(k+1)}^\ell$.
      
    \item Apply~(\ref{E_38}) to update the critic weights.
      
    \item Terminate upon convergence of $\Vert	\hat S^{\ell+1}(\dots) - \hat S^{\ell}(\dots)\Vert $.
    \end{enumerate}
  \end{center}

\section{Results and Discussion}
\label{sec:results-discussion}

To validate the model-free adaptive controller, it is applied to a flexible wing aircraft model. Two case studies are simulated. They are designed to assess the convergence of the learning process and the controller's stability and robustness in the absence of a dynamic model and in the face of various disturbance levels.

The flexible wing airplane model experimented in~\cite{cook_spottiswoode_2005} is used to generate the necessary measurements at a trim speed of \SI{10.8}{\m\per\s}. This value is based on experimental data collected by Kilkenny in~\cite{Kilkenny_1986} to construct a flight envelope at different speeds to obtain valid aerodynamic data for a flexible wing system. Table~\ref{tab:configuration-parameters} shows the configuration parameters of the Hiway Demon hang glider adopted in this manuscript. Since this work is concerned with the automatic control of such types of aerial vehicles, the pilot may be regarded as a cargo mass, for instance. 

\begin{table}[htb]
  \centering
  \caption{Hang Glider (Hiway Demon) Configuration Parameters}
  \label{tab:configuration-parameters}
  \begin{tabular}{ll}
    \toprule 
    Parameter & Value \\
    \midrule 
    pilot mass & \SI{80}{\kg} \\
    wing mass &  \SI{31}{\kg} \\
    wing area & \SI{16.26}{\m\squared} \\
    wing span & \SI{10}{\m}  \\
    reference chord length & \SI{1.626}{\m}  \\
    hang point position & \SI{0.04}{\m}  \\
    hang strap length & \SI{1.2}{\m}  \\
    control frame position & \SI{0.06}{\m}  \\
    control frame height & \SI{1.65}{\m}  \\
    distance between the pilot's hands on the control bar& \SI{0.3}{\m} \\
    \bottomrule
  \end{tabular}
\end{table}

The discrete-time state space matrices of the decoupled longitudinal and lateral planes are given by
\begin{gather*}
  {\boldsymbol{A}}^{Lon} =
\begin{bmatrix*}[r]
0.9982  &  0.0065  &  0.0012 &  -0.0971\\
-0.0139  &  0.9774  &  0.1055 &   0.0136\\
0.0027  & -0.0043  &  0.9858 &  -0.0002\\
0.0000  & -0.0000  &  0.0099 &   1.0000
\end{bmatrix*} 
\\
{\boldsymbol{A}}^{Lat} =
\begin{bmatrix*}[r]
0.9977  & -0.0028  & -0.1069  &  0.0971  & -0.0131\\
-0.0131  &  0.8092  &  0.0677  & -0.0007  &  0.0001\\
0.0026  &  0.0332  &  0.9802  &  0.0001  & -0.0000\\
-0.0001  &  0.0090  &  0.0004  &  1.0000  &  0.0000\\
0.0000  &  0.0002  &  0.0099  &  0.0000  &  1.0000
\end{bmatrix*}
\\
{\boldsymbol{B}}^{Lon} =
\begin{bmatrix*}[r]
0.0000\\
0.0040\\
0.0741\\
0.0004
\end{bmatrix*}
~,~
{\boldsymbol{B}}^{Lat} =
\begin{bmatrix*}[r]
-0.0003\\
0.0327\\
0.0049\\
0.0002\\
0.0000
\end{bmatrix*} 
\end{gather*}
with initial conditions
$\boldsymbol{Z}_0^{Lon} = \mqty[28 & -1.0 & -0.6 & 1.0]^T$
and
$\boldsymbol{Z}_0^{Lat} = \mqty[10 & 0.9  & 0.9 &  1.0 & -0.5]^T$. For a more realistic setup the control signals $\boldsymbol{u}^{Lon}=\alpha$ and $\boldsymbol{u}^{Lat}=\beta$ are bound to $\pm \SI[parse-numbers=false]{\pi/3}{\radian}$.

The weighting matrices are set to
\begin{align*}
  {\boldsymbol{Q}}^{Lon}
  & =
    \begin{bmatrix*}[r]
      0.0006  &  0.0400  &  1.0000  &  1.0000
    \end{bmatrix*} 
  \\
  {\boldsymbol{Q}}^{Lat}
  & =
    \begin{bmatrix*}[r]
      0.0006  &  0.2500  &  0.2500  &  1.0000  &  1.0000
    \end{bmatrix*} 
  \\
  {R}^{Lon}
  & =0.9803 ~,~
    {R}^{Lat} = 0.9803 
\end{align*}
The actor weights are initialized to
\begin{align*}
  {\boldsymbol{W}}_a^{Lon} & = \mqty[0.0317 & 0.0014 & -2.4171 & -3.0740]^T
  \\
  {\boldsymbol{W}}_a^{Lat} & = \mqty[0.0404 & -0.6407 & -2.2064 & -1.8473 & -1.8872]^T.
\end{align*}

The critic weights are initialized to positive definite matrices, 
${\boldsymbol{W}}_c^{Lon} =10 \, I_{5 \times 5}, \, {\boldsymbol{W}}_c^{Lat} = 10 \, I_{6 \times 6}$.
The actor and critic learning rates are fixed to $\eta_{{W}_a} = \eta_{{W}_c} = 0.001$.
The simulations are performed in Matlab~2017 on a server with 16 virtual CPUs and 48~GB of memory. The sampling period is taken as~\SI{0.01}{\second}.

\subsection{Case Study 1}
In the first case study, the aforementioned longitudinal and lateral state space models are used to simulate the hang glider. The system's open- and closed-loop poles are listed in Table~\ref{tab:eigvwo}. As can be noticed, both the longitudinal and lateral open-loop dynamics are unstable. Integrating the proposed controller in a feedback loop asymptotically stabilized the system by shifting the poles to a stable region.

The optimal control gains for longitudinal and lateral directions are given as follows:
	\begin{align*}
	{\boldsymbol{W}}_a^{Lon} & = \mqty[0.5229  & -0.9582   &-2.6512  & -2.5554]^T,
	\\
	{\boldsymbol{W}}_a^{Lat} & = \mqty[0.0120  & -0.9219  & -2.4250  & -0.9458  & -1.1173]^T.
	\end{align*}
The evolution of the closed-loop poles during the learning process is depicted in Figure~\ref{fig:eigwo}.
The figure shows the open-loop poles with a {\color{red}  o};
the final closed-loop poles, resulting from the model-free control, with a {\color{blue}$\times$}; and
the closed-loop poles during the learning process with {\color{green}  $\bullet$}.
The pole search space covered both stable and unstable regions before eventually converging to asymptotically stable poles once the actor and critic weights are settled to their final values. 
The evolution of the adaptive weights of the actor and critic neural networks is demonstrated in Figure~\ref{fig:actcrtwo}. After an initial fluctuation period, all the weights converged to their steady-state values, as proven in Theorem~\ref{thm:stability-convergence}.
This is also illustrated in~Figure~\ref{fig:dynwo}, which demonstrates the dynamics and control signals of the longitudinal and lateral subsystems. It took about (\SI{15}{\second}) and (\SI{20}{\second}) for the longitudinal and lateral states to decay to zero respectively. It is worth noticing that the controller started with an aggressive approach in an attempt to stabilize the system as fast as possible. This is clear from the  control signals which reached their saturation values of $\pm \SI[parse-numbers=false]{\pi/3}{\radian}$ three times in the first (\SI{5}{\second}) and (\SI{10}{\second}) for the longitudinal and lateral systems respectively.

\begin{table}[htb]
	\centering
	\caption{Open- and closed-loop poles of the decoupled systems}
	\label{tab:eigvwo}
	\begin{tabular*}{20pc}{@{\extracolsep{\fill}}lll@{}}%
		\toprule
		\textbf{The longitudinal system:} &
		\\
		\midrule
		Open-loop poles
		&  $0.9801 \, e^{\pm 0.0219}$
		\\
		&  $1.0009 \, e^{\pm 0.0116}$
		\\
		\midrule
		Closed-loop poles
		&  $0.8771, \, 0.8884$
		\\
		&  $0.9975 \, e^{\pm  0.0123}$
		\\
		\midrule
		\textbf{The lateral system:} & 
		\\
		\midrule
		Open-loop poles
		&  $0.7978, \, 0.9949$
		\\
		&  $0.9973 \, e^{\pm 0.0088}$
		\\
		&  $1.0000$
		\\
		\midrule
		Closed-loop poles
		&  $0.7825, \, 0.9727$
		\\
		&  $0.9960$
		\\
		&  $0.9969 \, e^{\pm 0.0082}$
		\\
		\hline
	\end{tabular*}
\end{table}

\subsection{Case Study 2}
In this scenario, the robustness of the proposed approach is tested against unmodeled aerodynamics and uncertainties. Since, the aerodynamic models of the flexible wing aircraft are unknown, artificial nonlinearities are simulated by intentionally adding a degree of randomness (for each entry of $A, \, B,$ and $Z_k$) around the nominal dynamical parameters (i.e., the trim flight condition).
The different time-dependent white-noise disturbances and uncertainties are computed at every time step (i.e., evaluation step) from a normal distribution $\mathcal{N}(0,\,1)$. The uncertainty levels in the aerodynamics and the states are scaled up to $\pm 50\%$ and $\pm 20\%$ of their respective nominal values.
The noise in the states is very useful in a real-world scenario as it reflects the typical noise in the feedback data measurements introduced by the sensors.

This simulation scenario employs a combined control scheme, where the overall optimal control signal is divided into parts; First, the Riccati approach is allowed to find the optimal control signal using the nominal aerodynamic parameters (i.e., nominal undisturbed aerodynamics at the trim speed). Second, the learning system is integrated to support the optimal online control decision for uncertainties. 
The Riccati solution assumes a perfect knowledge of the system's model. Hence, in this case, the adaptive model-free controller can be seen a regulatory controller that compensates for the irregularities that are missed by the Riccati control signal. Hence, the resulting systems will have following form
\begin{eqnarray*}
	\boldsymbol {Z}_{k+1}=(\boldsymbol{A}+\boldsymbol{\Delta A}_k) (\boldsymbol{Z}_{k}+\boldsymbol{\Delta Z}_{k})+\left(\boldsymbol {B} \, \boldsymbol{u}^{R}_{k}+\Delta \boldsymbol{B}_k \, \boldsymbol{u}^{MF}_{k}\right), 
	\label{E_39}
\end{eqnarray*}
where $\boldsymbol{\Delta Z}_k,$ $\boldsymbol{\Delta A}_k,$ and $\Delta \boldsymbol{B}_k$ are the induced randomnesses in each entry of the states $\boldsymbol{Z}_k,$ system matrix $\boldsymbol{A}$, and control input matrix $\boldsymbol{B}$ respectively. $u_k^{R}$ is the control input evaluated by the Riccati control solution, while $u^{MF}_k$ is the control signal generated by the model-free controller.

Later on, the control law obtained by the model-free learning process (i.e., $u^{MF}_k$) is compared with other situations as if it is applied solely to the system (i.e., without the supporting Riccati regulation part), which simply means coming back to the first scenario, where the control signal is decided completely by the learning process.

The discrete-time state space matrices of the decoupled longitudinal and lateral frames at last evaluation instance \textit{k} are given by 
	\begin{gather*}
	{\boldsymbol{A}}^{Lon}_k =
	\begin{bmatrix*}[c]
	0.9983  &  0.0055 &  0.0014  & -0.0907\\
	-0.0145  &  0.9871 &  0.1127  &  0.0108\\
	0.0019  & -0.0034 &  0.9864  & -0.0001\\
	0  & 0 &  0.0091  &  1.0000
	\end{bmatrix*},
	\\
	{\boldsymbol{A}}^{Lat}_k =
	\begin{bmatrix*}[c]
	0.9971 &  -0.0028 &  -0.0987  &  0.1106 &  -0.0126\\
	-0.0085 &   0.8307 &   0.0734  & -0.0005 &   0.0001\\
	0.0013 &   0.0338 &   0.9860  &  0.0001 &  0\\
	0 &   0.0076 &   0.0003  &  1.0000 &   0\\
	0 &   0.0002 &   0.0092  &  0 &   1.0000
	\end{bmatrix*},
	\\
	{\boldsymbol{B}}^{Lon}_k =
	\begin{bmatrix*}[c]
	0\\
	0.0031\\
	0.0544\\
	0.0003\\
	\end{bmatrix*}
	~,~
	{\boldsymbol{B}}^{Lat}_k =
	\begin{bmatrix*}[c]
	-0.0002\\
	0.0427\\
	0.0030\\
	0.0002\\
	0
	\end{bmatrix*} .
	\end{gather*}
The final optimal control gains for the longitudinal and lateral directions, recorded by the learning process at last evaluation instance \textit{k}, are given as follows:
\begin{align*}
{\boldsymbol{W}}_a^{Lon} & = \mqty[0.1961  & -0.0554 &  -2.4252 &  -3.0799]^T,
\\
{\boldsymbol{W}}_a^{Lat} & = \mqty[-0.0038 &  -0.6563 &  -2.3076 &  -1.9291 &  -2.0037]^T.
\end{align*}

In this scenario, the different eigenvalues which are evaluated using the disturbed system ($A+\Delta A_k$ and $B+\Delta B_k$) at instance \textit{k} are represented graphically as follows; The open-loop poles of the disturbed system are given the {\color{red} o} marks, while the notations {\color{blue} o} refer to the final closed-loop poles of the disturbed system employing Riccati control gains calculated using ($A+\Delta A_k$ and $B+\Delta B_k$). The closed-loop poles (represented by {\color{blue}$\times$}) are the result of employing the combined model-free approach and the Ricatti control solution using nominal system ($A$ and $B$). The notations {\color{red}$\times$} indicate the closed-loop poles using the model-free control gains only. The notations {\color{green} $\bullet$} and {\color{red} $\bullet$} represent the spectrum of the closed-loop poles using the learning process and Riccati approach separately. Finally, the symbols {\color{blue} $\bullet$} denote the spectrum of the closed-loop poles of the combined learning-Riccati process.

As listed in Tables~\ref{tab:eigwlon}  \& \ref{tab:eigwlat}, the open-loop poles of the disturbed longitudinal and lateral subsystems at last evaluation step \textit{k} are unstable. Tables~\ref{tab:eigwlon}~\&~\ref{tab:eigwlat} and Figures~\ref{fig:eigw}~\&~\ref{fig:eig_new} show the closed-loop poles of the longitudinal and lateral dynamics. Figure~\ref{fig:eig_new} highlights the ability of the online stand-alone learning process to explore more stable areas, which resulted in faster dominant modes compared to those achieved using Riccati or combined Riccati-learning processes.
It is worth noticing how the learning instances (poles) in Figure~\ref{fig:eigw} are denser compared to Figure~\ref{fig:eigwo}. This is due to the longer time taken by the actor and critic weights to converge to their steady-state values, as can be witnessed from Figure~\ref{fig:actcrtw}. The final values of the closed-loop poles for the longitudinal and lateral systems listed in Tables~\ref{tab:eigwlon}  \& \ref{tab:eigwlat} and hence shown in Figure~\ref{fig:eigw} emphasized the superiority of the proposed learning approach even when it is working in stand-alone mode. The dominant modes of the stand-alone learning controller are asymptotically stabilized and became much faster compared to all other situations. The spectrum of the closed-loop poles evolution calculated by stand-alone Riccati approach, combined Riccati-learning controller, and stand-alone learning control system, exhibited by Figure~\ref{fig:eig_new}, assure that whenever the learning process exists, it widens the stability region and allows the system to take more quicker decisions, even in highly and continuously disturbed aerodynamic environment.

The intelligent adaptive controller stabilized the hang glider by again forcing the unstable poles to slide to an asymptotically stable region. 
Figure~\ref{fig:dynw} demonstrates the asymptotic stability behavior of the decoupled dynamical systems. The longitudinal and lateral control input signals and their associated dynamics took almost (\SI{3}{\second}) and (\SI{6}{\second}) to practically vanish. Despite the relatively significant effort the controller had to apply to stabilize the lateral dynamics, the stabilization of the longitudinal dynamics were much easier to achieve. This is clear from the span of both control signals. This fast convergence and relative control comfort is due to the synergistic integration of the Riccati control signal, stemming from the nominal model, and the model-free control policy compensating for the imperfections in that model. 
Such behavior is tested by disabling the Riccati control component. The results are revealed in Figure~\ref{fig:dyn_withnoise_woriccati}. This time, it took the longitudinal and lateral dynamics more effort and a longer time to stabilize (about (\SI{10}{\second}) and (\SI{15}{\second}), respectively). 
It demonstrates that despite the time-variant noise in the system dynamics, the intelligent model-free controller is able to control the aircraft even without the help of a Riccati controller. 
\begin{table}
    \centering
    \caption{Open- and closed-loop poles of the longitudinal disturbed system \\ at the last evaluation step \textit{k}}
    \label{tab:eigwlon}
    \begin{tabular*}{20pc}{@{\extracolsep{\fill}}lll@{}}%
        \toprule
        \textbf{The longitudinal system:} &
        \\
        \midrule
        Open-loop poles of the disturbed system
        &  $0.9859 \,  e^{\pm 0.0200}$
        \\
        $(\boldsymbol{A}+\Delta \boldsymbol{A}^{Lon}_k) \, \& \,(\boldsymbol{B}+\Delta \boldsymbol{B}^{Lon}_k)$
        &  $1.0002 \,  e^{\pm 0.0102}$
        \\
        \midrule
        Closed-loop poles using Riccati approach;
        &  $0.9634 \,  e^{\pm 0.0090}$
        \\
        The control gains were calculated using 
        &  $0.9964 \,  e^{\pm 0.0089}$
        \\
        system: $(\boldsymbol{A}\, \& \,\boldsymbol{B})$ &
        \\
        \midrule
        Closed-loop poles using Riccati approach;
        &  $0.9623 \,  e^{\pm 0.0130}$ 
        \\
        The control gains were calculated using 
        &  $0.9967 \,  e^{\pm 0.0093}$
        \\
        system: $(\boldsymbol{A}+\Delta \boldsymbol{A}^{Lon}_k) \, \& \,(\boldsymbol{B}+\Delta \boldsymbol{B}^{Lon}_k)$ & \\ 
        \midrule
        Closed-loop poles using the combined  
        &  $0.9781 \,  e^{\pm 0.0178}$
        \\ Riccati-model-free approach
        &  $0.9964 \,  e^{\pm 0.0065}$
        \\
        \midrule
        Closed-loop poles using the model-free
        &  $0.8713,  \, 0.9751$
        \\ approach
        &  $0.9963 \,  e^{\pm 0.0100}$
        \\
        \hline
    \end{tabular*}
\end{table} 
\begin{table}
	\centering
	\caption{Open- and closed-loop poles of the lateral disturbed system at the last\\  evaluation step \textit{k}}
	\label{tab:eigwlat}
	\begin{tabular*}{20pc}{@{\extracolsep{\fill}}lll@{}}%
		\toprule
		\textbf{The lateral system:} &
		\\
		\midrule
		Open-loop poles of the disturbed system
		&  $0.8167, \, 0.9947, \, 1$
		\\
		$(\boldsymbol{A}+\Delta \boldsymbol{A}^{Lat}_k) \, \& \,(\boldsymbol{B}+\Delta \boldsymbol{B}^{Lat}_k)$
		&  $1.0012 \,  e^{\pm 0.0037}$
		\\
		\midrule
		Closed-loop poles using Riccati approach;
		&  $0.8165$ 
		\\
		The control gains were calculated using 
		&  $0.9948 \, e^{\pm 0.0046}$
		\\
		system: $(\boldsymbol{A}\, \& \,\boldsymbol{B})$ &
		$0.9976\, e^{\pm 0.0085}$
		\\
		\midrule
		Closed-loop poles using Riccati approach;
		&  $0.8156$
		\\
		The control gains were calculated using 
		&  $0.9914 \,  e^{\pm 0.0057}$
		\\
		system: $(\boldsymbol{A}+\Delta \boldsymbol{A}^{Lat}_k) \, \& \,(\boldsymbol{B}+\Delta \boldsymbol{B}^{Lat}_k)$ & 
		$0.9976 \,  e^{\pm 0.0055}$
		\\ 
		\midrule
		Closed-loop poles using the combined  
		&  $0.8152$
		\\ Riccati-model-free approach
		&  $0.9947 \,  e^{\pm 0.0047}$
		\\
		& $0.9974 \,  e^{\pm 0.0082}$
		\\
		\midrule
		Closed-loop poles using the model-free
		&  $0.8105, \, 0.9824$
		\\ approach
		&  $0.9923$ 
		\\
		&$0.9967 \,  e^{\pm 0.0070}$
		\\
		\hline
	\end{tabular*}
\end{table}

\begin{figure*}[htb]
	\centering
	\includegraphics[width=0.9\linewidth]{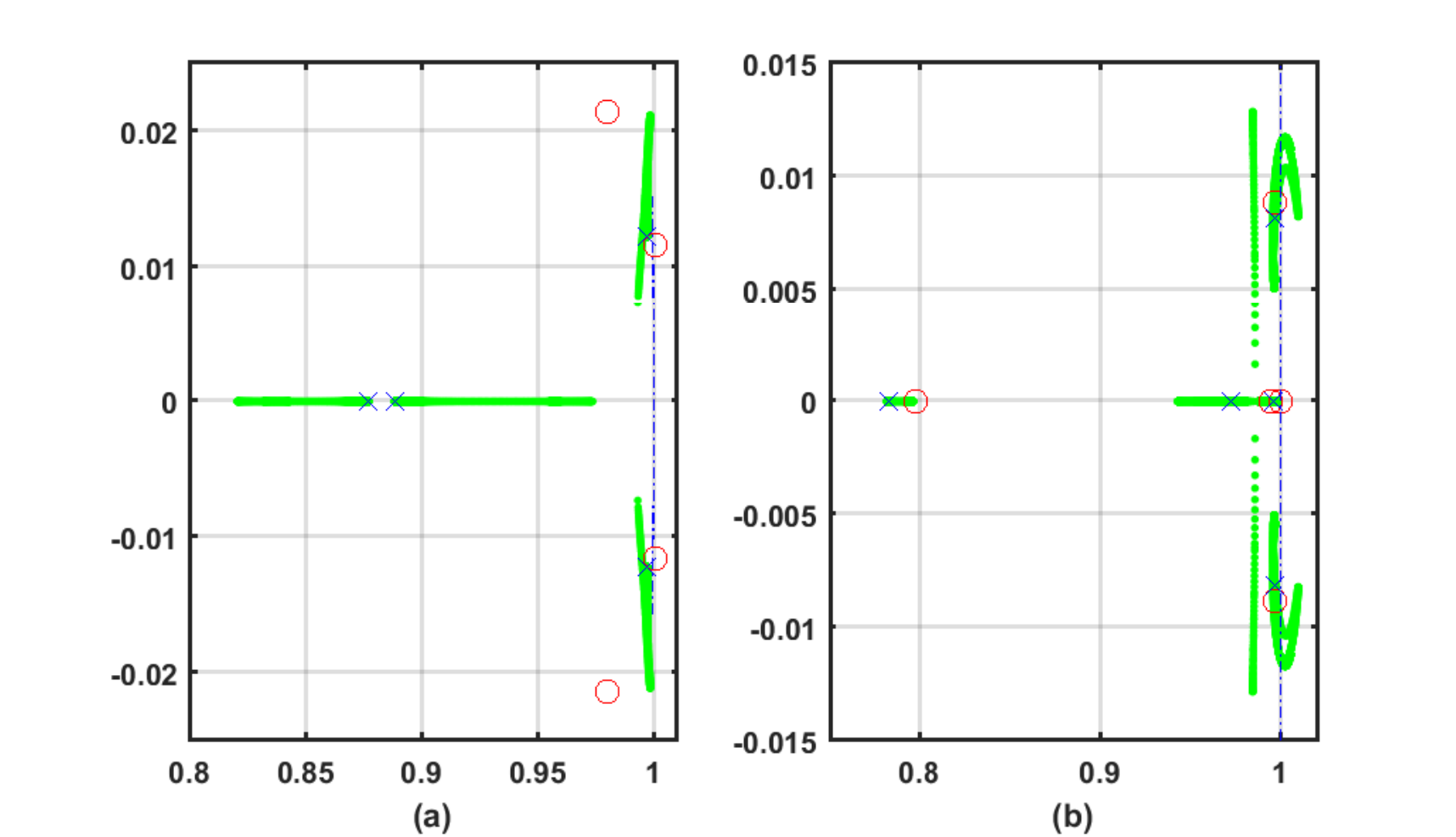}
	\caption{Evolution of the closed-loop poles during learning: (a)~longitudinal and~(b) lateral dynamics.}
	\label{fig:eigwo}
\end{figure*}

\begin{figure*}[htb]
	\centering
	\includegraphics[width=0.9\linewidth]{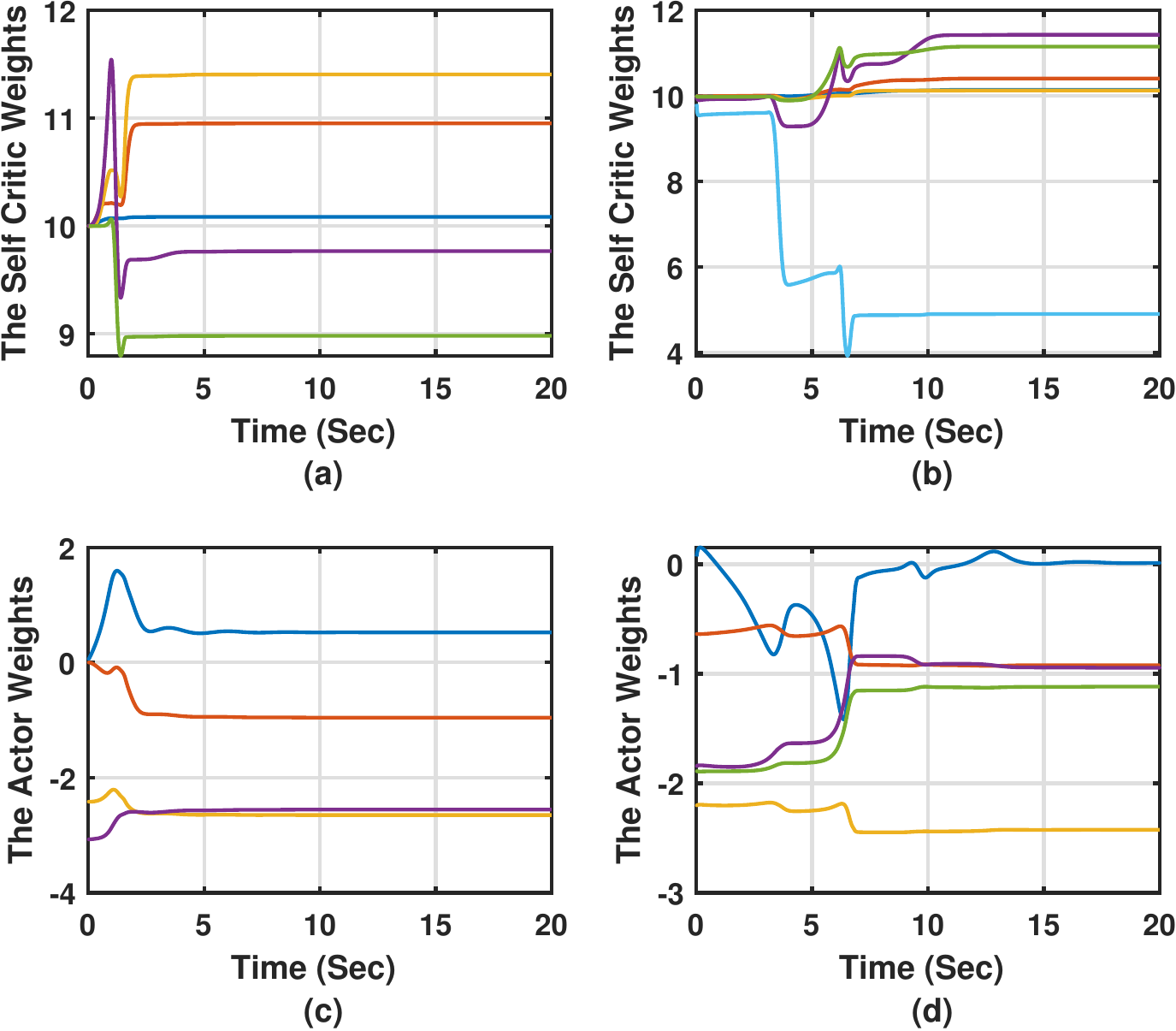}
	\caption{Actor-critic weights: (a),(c)~longitudinal and (b),(d)~lateral dynamics.}
	\label{fig:actcrtwo}
\end{figure*}

\begin{figure*}[htb]
	\centering
	\includegraphics[width=0.7\linewidth]{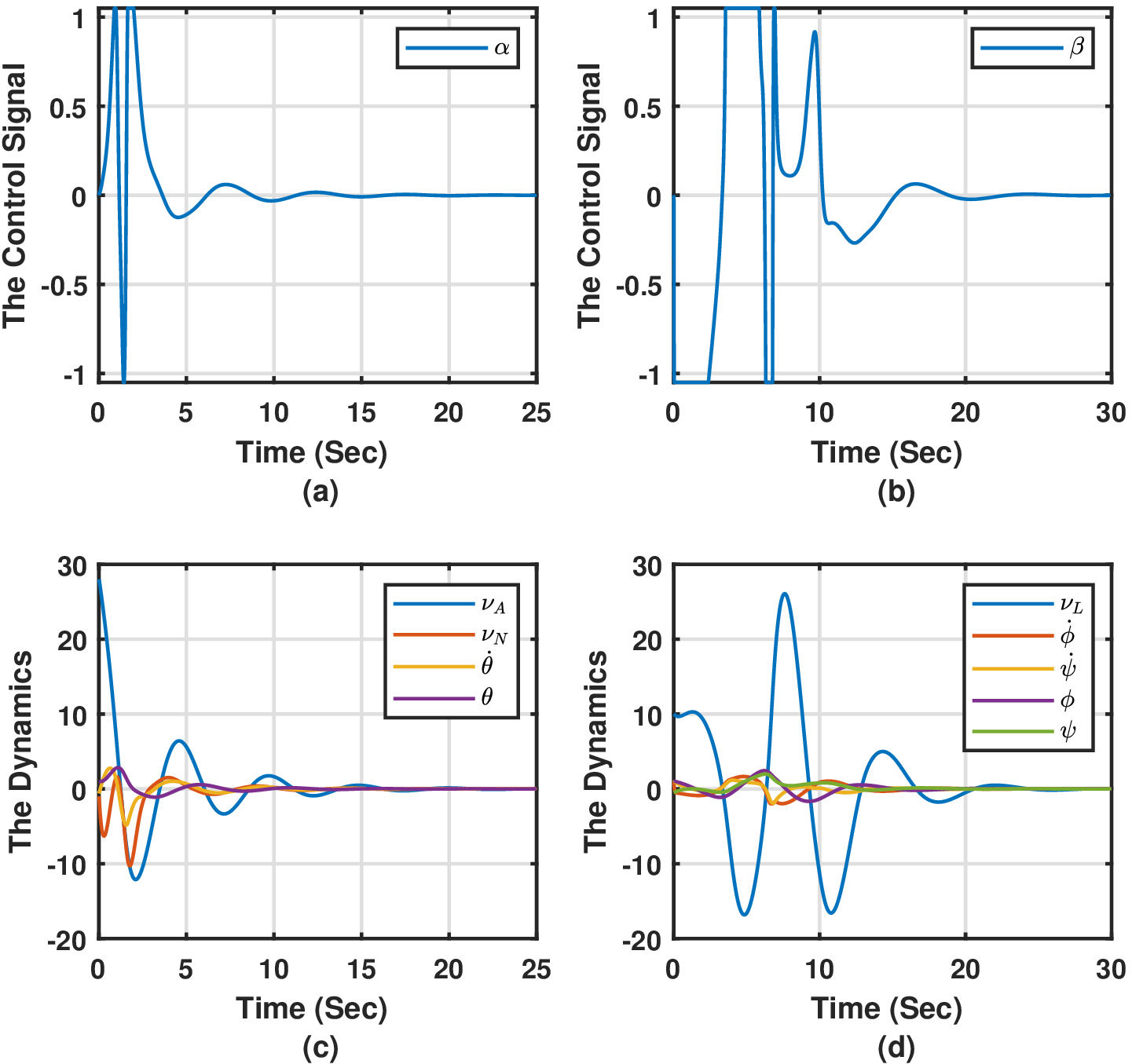}
	\caption{The control signals and dynamics of (a),(c):~longitudinal; (b),(d):~lateral. All the signals are in their respective SI units.} 
	\label{fig:dynwo}
\end{figure*}

\begin{figure*}[htb]
  \centering
  \includegraphics[width=0.7\linewidth]{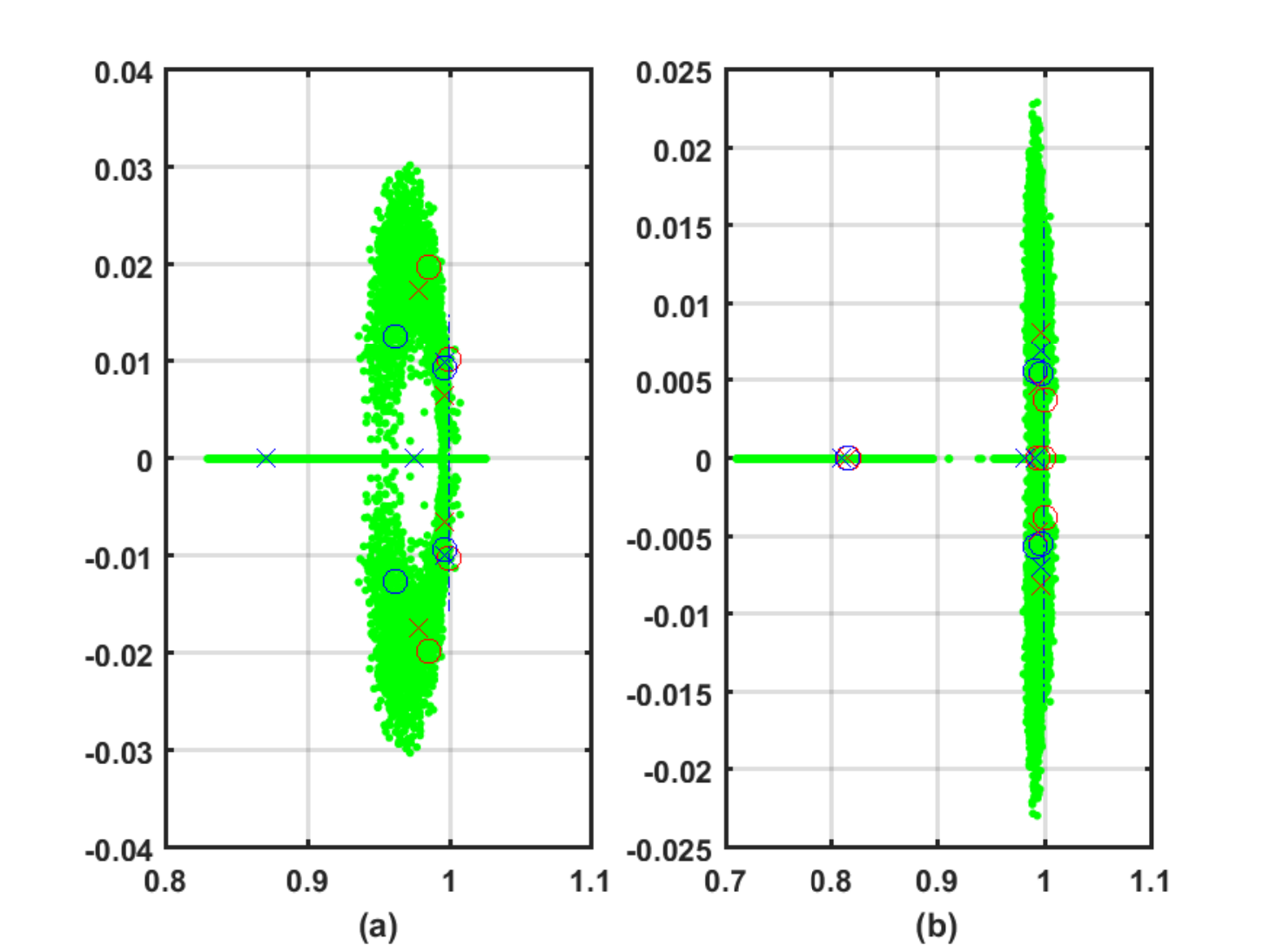}
  \caption{Evolution of the closed-loop poles during learning of the disturbed system: (a)~longitudinal and~(b) lateral dynamics.}
  \label{fig:eigw}
\end{figure*}

\begin{figure*}[htb]
	\centering
	\includegraphics[width=0.8\linewidth]{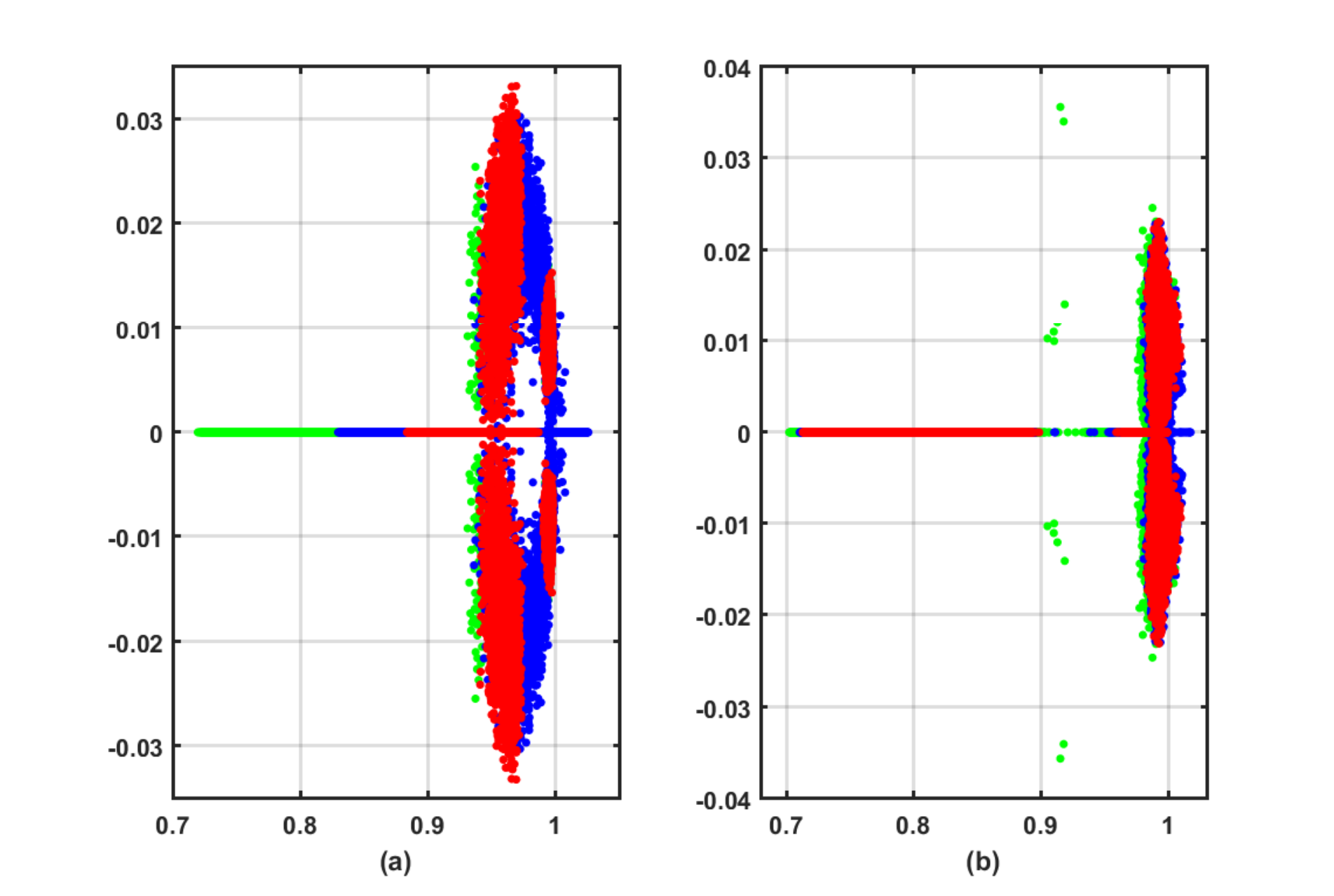}
	\caption{Evolution of the closed-loop poles during learning of the disturbed system using the stand-alone adaptive learning scheme, Riccati approach, combined Riccati-learning mechanism: (a)~longitudinal and~(b) lateral dynamics.}
	\label{fig:eig_new}
\end{figure*}

\begin{figure*}[htb]
  \centering
  \includegraphics[width=0.8\linewidth]{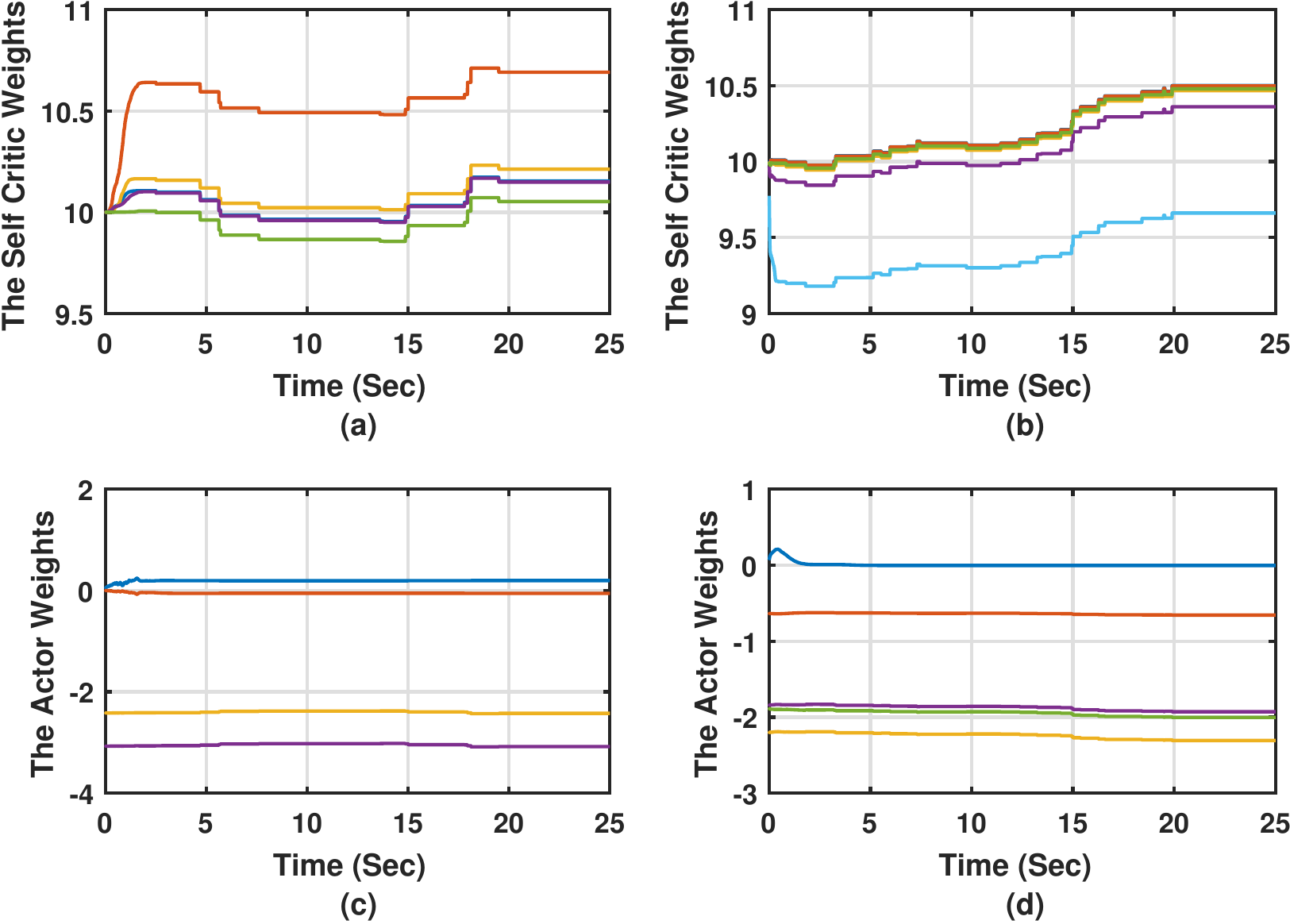}
  \caption{Actor-critic weights for the disturbed system: (a),(c)~longitudinal and (b),(d)~lateral dynamics.}
  \label{fig:actcrtw}
\end{figure*}

\begin{figure*}[htb]
  \centering
  \includegraphics[width=0.8\linewidth]{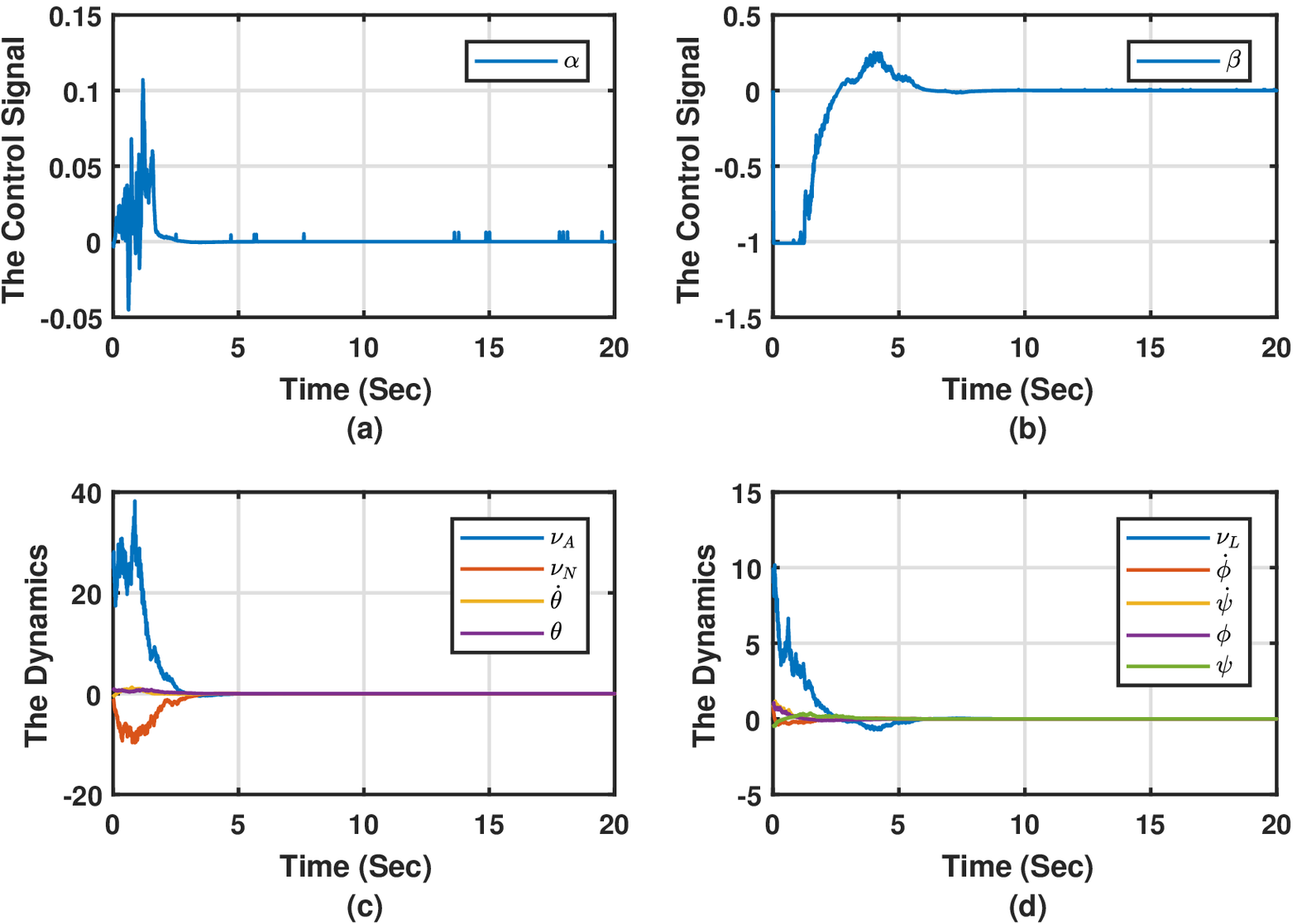}
  \caption{The control signal and dynamics of the disturbed system. (a),(c):~longitudinal; (b),(d):~lateral. All the signals are in their respective SI units.} 
  \label{fig:dynw}
\end{figure*}

\begin{figure*}[htb]
	\centering
	\includegraphics[width=0.8\linewidth]{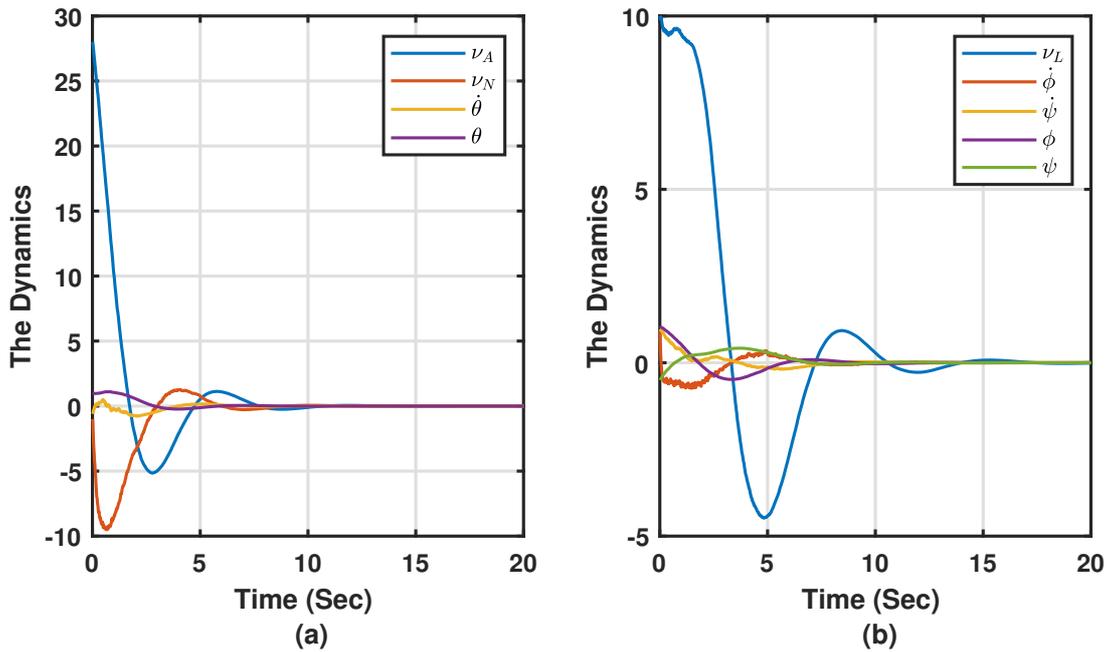}
  \caption{The dynamics of the disturbed system (without Riccati control component). (a):~longitudinal; (b):~lateral. All the signals are in their respective SI units.} 
	\label{fig:dyn_withnoise_woriccati}
\end{figure*}
\clearpage
\pagebreak

\section{Conclusion}
\label{sec:conclusion}
The automatic control of a flexible wing aircraft was addressed. It is challenging due to the lack of an accurate dynamical model faithfully describing the system dynamics. This work introduced an online adaptive learning controller to accomplish this task. It is based on machine learning and optimal control methods to generate model-free control decisions. Unlike traditional classical approaches, the proposed controller does not use or require any information about the system's dynamics. Instead, it finds optimal control approximations in real-time using few feedback signals that can be measured through sensors mounted aboard of the hang glider. The solution is implemented by employing a pair of actor-critic neural networks where the weights are adapted online through a gradient descent technique. The feedback loop is proven to be asymptotically stable while the reinforcement learning algorithm is guaranteed to converge. A Riccati solution framework is shown to be equivalent to solving the underlying model-free Bellman optimality equation. The validation results show the superior performance of the developed adaptive learning controller in terms of the stability and convergence even when the system is subject to significant time-variant disturbances.

\vspace{-300pt}
\bibliographystyle{Bib/iet}
\bibliography{Bib/mybibliography}

\begin{thebibliography}{10}

\bibitem{cook_spottiswoode_2005}
Cook, M.V., Spottiswoode, M.: `Modelling the flight dynamics of the hang
  glider', \emph{The Aeronautical Journal},  2005, \textbf{109}, (1102),
  pp.~I--XX

\bibitem{Kilkenny_1983}
Kilkenny, E.A.
\newblock `{An evaluation of a mobile aerodynamic test facility for hang glider
  wings}'.
\newblock (College of Aeronautics, Cranfield Institute of Technology,  1983.
  8330

\bibitem{Kilkenny_1984}
Kilkenny, E.
\newblock `Full scale wind tunnel tests on hang glider pilots'.
\newblock (Cranfield Institute of Technology, College of Aeronautics,
  Department of Aerodynamics,  1984.

\bibitem{Cook_Kilkenny_1986}
Cook, M.V., Kilkenny, E.A.
\newblock `An experimental investigation of the aerodynamics of the hang
  glider'.
\newblock In: Proceedings of the International Conference on Aerodynamics. (,
  1986. pp.~ 1--10

\bibitem{Kilkenny_1986}
Kilkenny, E.A.
\newblock `An experimental study of the longitudinal aerodynamic and static
  stability characteristics of hang gliders' [phdthesis].
\newblock Cranfield University,  1986

\bibitem{Blake_1991}
Blake, D.
\newblock `Modelling The Aerodynamics, Stability and Control of The Hang
  Glider' [mathesis].
\newblock College of Aeronautics, Cranfield Institute of Technology,  1991

\bibitem{cook_1994}
Cook, M.V.: `The theory of the longitudinal static stability of the
  hang-glider', \emph{The Aeronautical Journal},  1994, \textbf{98}, (978),
  pp.~292--304

\bibitem{Ochi_2017}
Ochi, Y.
\newblock `Modeling of flight dynamics and pilot's handling of a hang glider'.
\newblock In: AIAA Modeling and Simulation Technologies Conference. (American
  Institute of Aeronautics and Astronautics,  2017. pp.~ 1758--1776

\bibitem{Sweeting1981}
Sweeting, J.
\newblock `An Experimental Investigation of Hang Glider Stability' [mathesis].
\newblock College of Aeronautics, Cranfield University,  1981

\bibitem{Cook_2013}
Cook, M.: `Flight Dynamics Principles: A Linear Systems Approach to Aircraft
  Stability and Control'.
\newblock 3rd ed. Aerospace Engineering. (Butterworth-Heinemann,  2012)

\bibitem{Kroo_1983}
Kroo, I.
\newblock `Aerodynamics, Aeroelasticity and Stability of Hang Gliders'
  [phdthesis].
\newblock Stanford University,  1983

\bibitem{Powton_1995}
Powton, J.
\newblock `A Theoretical Study of the Non-linear Aerodynamic Pitching Moment
  Characteristics of the Hang Glider and its Influence on Stability and
  Control' [mathesis].
\newblock College of Aeronautics, Cranfield Institute of Technology,  1995

\bibitem{DE_MATTEIS_1990}
De.Matteis, G.: `Response of hang gliders to control', \emph{The Aeronautical
  Journal},  1990, \textbf{94}, (938), pp.~289--294

\bibitem{De_Matteis_1991}
de~Matteis, G.: `Dynamics of hang gliders', \emph{Journal of Guidance Control
  and Dynamics},  1991, \textbf{14}, (6), pp.~1145--1152

\bibitem{Spottiswoode_2001}
Spottiswoode, M.
\newblock `A Theoretical Study of the Lateral-directional Dynamics, Stability
  and Control of the Hang Glider' [mathesis].
\newblock College of Aeronautics, Cranfield Institute of Technology,  2001

\bibitem{Ochi_2015}
Ochi, Y.
\newblock `Modeling of the longitudinal dynamics of a hang glider'.
\newblock In: AIAA Modeling and Simulation Technologies Conference. (American
  Institute of Aeronautics and Astronautics,  2015. pp.~ 1591--1608

\bibitem{Rollins_2000}
Rollins, R.
\newblock `Study of Experimental Data to Assess the Longitudinal Stability and
  Control of the Hang Glider' [mathesis].
\newblock College of Aeronautics, Cranfield University,  2000

\bibitem{Aguilar17}
Aguilar.Ibañez, C.: `Stabilization of the pvtol aircraft based on a sliding
  mode and a saturation function', \emph{International Journal of Robust and
  Nonlinear Control},  2017, \textbf{27}, (5), pp.~843--859

\bibitem{Rubio18}
d.~J..{Rubio}, J., {Pieper}, J., {Meda-Campaña}, J.A., {Aguilar}, A.,
  {Rangel}, V.I., {Gutierrez}, G.J.: `Modelling and regulation of two
  mechanical systems', \emph{IET Science, Measurement Technology},  2018,
  \textbf{12}, (5), pp.~657--665

\bibitem{RUBIO2018155}
Rubio, J.: `Robust feedback linearization for nonlinear processes control',
  \emph{ISA Transactions},  2018, \textbf{74}, pp.~155 -- 164

\bibitem{Howard_1960}
Howard, R.A.: `Dynamic Programming and Markov Processes'.
\newblock Four volumes. (Cambridge. MA: MIT Press,  1960)

\bibitem{Werbos1992}
Werbos, P.
\newblock 13.
\newblock In: White, D.A., Sorge, D.A., editors. `Approximate dynamic
  programming for real-time control and neural modeling'. (New York: Van
  Nostrand Reinhold: Van Nostrand Reinhold, New York,  1992. pp.~ 493--525

\bibitem{Bert_1995}
Bertsekas, D.P., Tsitsiklis, J.N.
\newblock `Neuro-dynamic programming: An overview'.
\newblock In: Proceedings of the IEEE Conference on Decision and Control.
  vol.~1. (,  1995. pp.~ 560--564

\bibitem{Werbos1990}
Miller, W.T., Sutton, R.S., Werbos, P.J.: `Neural Networks for Control: A Menu
  of Designs For Reinforcement Learning Over Time'.
\newblock 1st ed. (MIT Press,  1990)

\bibitem{Sutton_1998}
Sutton, R.S., Barto, A.G.: `Reinforcement Learning: An Introduction'.
\newblock (Massachusetts USA: MIT Press,  1998)

\bibitem{AbouheafRV2017}
Abouheaf, M., Gueaieb, W.
\newblock `Multi-agent reinforcement learning approach based on reduced value
  function approximations'.
\newblock In: IEEE International Symposium on Robotics and Intelligent Sensors
  (IRIS). (,  2017. pp.~ 111--116

\bibitem{AbouheaIJCNN2013}
Abouheaf, M., Lewis, F.
\newblock `Approximate dynamic programming solutions of multi-agent graphical
  games using actor-critic network structures'.
\newblock In: International Joint Conference on Neural Networks (IJCNN). (,
  2013. pp.~ 1--8

\bibitem{AbouheafCH2014}
Abouheaf, M., Lewis, F.
\newblock Chapter 1.
\newblock In: Liu, D., Alippi, C., Zhao, D., Zhang, H., editors. `Dynamic
  graphical games: Online adaptive learning solutions using approximate dynamic
  programming'. (World Scientific,  2014. pp.~ 1--48

\bibitem{Lewis_2012}
Lewis, F., Vrabie, D., Syrmos, V.: `Optimal Control'.
\newblock 3rd ed. (New York, USA: John Wiley,  2012)

\bibitem{Bellman1957}
Bellman, R.: `Dynamic Programming'.
\newblock (Princeton University Press,  1957)

\bibitem{Bryson1996}
Bryson, A.: `Optimal control-1950 to 1985', \emph{IEEE Control Systems},  1996,
  \textbf{16}, (3), pp.~26--33

\bibitem{Sen1999}
Sen, S., Weiss, G.
\newblock 6.
\newblock In: Weiss, G., editor. `Learning in multi-agent systems'. (Cambridge,
  MA, USA: MIT Press,  1999. pp.~ 259--298

\bibitem{Widrow1973}
Widrow, B., Gupta, N.K., Maitra, S.: `Punish/reward: Learning with a critic in
  adaptive threshold systems', \emph{IEEE Transactions on Systems, Man, and
  Cybernetics},  1973, \textbf{SMC-3}, (5), pp.~455--465

\bibitem{Werbos1989}
Werbos, P.
\newblock `Neural networks for control and system identification'.
\newblock In: Proceedings of the 28th IEEE Conference on Decision and Control.
  vol.~1. (,  1989. pp.~ 260--265

\bibitem{Werbos1974}
Werbos, P.
\newblock `Beyond Regression: New Tools for Prediction and Analysis in the
  Behavior Sciences' [phdthesis].
\newblock Harvard University,  1974

\bibitem{Busoniu2008}
Busoniu, L., Babuska, R., Schutter, B.D.: `A comprehensive survey of
  multi-agent reinforcement learning', \emph{IEEE Transactions on Systems, Man,
  and Cybernetics, Part C (Applications and Reviews)},  2008, \textbf{38}, (2),
  pp.~156--172

\bibitem{Vrancx2008}
Vrancx, P., Verbeeck, K., Nowe, A.: `Decentralized learning in markov games',
  \emph{IEEE Transactions on Systems, Man, and Cybernetics, Part B
  (Cybernetics)},  2008, \textbf{38}, (4), pp.~976--981

\bibitem{Tamimi2008}
Al.Tamimi, A., Lewis, F.L., Abu.Khalaf, M.: `Discrete-time nonlinear {HJB}
  solution using approximate dynamic programming: Convergence proof',
  \emph{IEEE Transactions on Systems, Man, and Cybernetics, {Part~B}
  (Cybernetics)},  2008, \textbf{38}, (4), pp.~943--949

\bibitem{Vrabie2009}
Vrabie, D., Pastravanu, O., Abu.Khalaf, M., Lewis, F.L.: `Adaptive optimal
  control for continuous-time linear systems based on policy iteration',
  \emph{Automatica},  2009, \textbf{45}, (2), pp.~477--484

\end{thebibliography}

\end{document}